\documentclass[11pt, a4paper]{article}

\usepackage[left=3cm,right=3cm,top=3cm,bottom=3cm,nohead]{geometry}
\usepackage{color}
\usepackage[round]{natbib}
\usepackage{mathrsfs}

\usepackage{amsthm}
\usepackage{enumerate}
\usepackage{multicol}
\usepackage{multirow}
\usepackage{amsmath, amssymb, bbm, bm}
\usepackage{amsfonts}
\usepackage{tikz}
\usetikzlibrary{shapes,arrows}

\usepackage{graphicx}

\newcommand{\ii}{{x}}
\newcommand{\jj}{{y}}

\newcommand{\noder}{\mbox{Re}}
\newcommand{\nodey}{\mbox{Y}}
\newcommand{\nodea}{\mbox{Ag}}
\newcommand{\nodev}{\mbox{Va}}
\newcommand{\nodec}{\mbox{Co}}
\newcommand{\nodeh}{\mbox{H}}

\newcommand{\G}{\mathcal{G}}
\newcommand{\X}{\mathfrak{X}}
\newcommand{\PP}{\mathbb{P}}

\DeclareMathOperator{\ing}{ing}

\DeclareMathOperator{\tail}{tail}
\DeclareMathOperator{\pa}{pa}
\DeclareMathOperator{\an}{an}
\DeclareMathOperator{\dec}{de}
\DeclareMathOperator{\barren}{barren}
\DeclareMathOperator{\dis}{dis}

\DeclareMathOperator{\mbl}{mb}

\newcommand\indep{\protect\mathpalette{\protect\independenT}{\perp}}
\def\independenT#1#2{\mathrel{\rlap{$#1#2$}\mkern2mu{#1#2}}}

\newtheoremstyle{break}
  {18pt}{9pt}{\itshape}
  {}
  {\bfseries}
  {.}
  {.5em}
  {}

\newtheoremstyle{breakdef}
  {12pt}{9pt}{}
  {}
  {\bfseries}
  {.}
  {.5em}
  {}

\theoremstyle{break}
\newtheorem{lemma}{Lemma}[section]
\newtheorem{theorem}[lemma]{Theorem}
\newtheorem{proposition}[lemma]{Proposition}

\theoremstyle{breakdef}
\newtheorem{definition}[lemma]{Definition}

\newtheorem{example}[lemma]{Example}

\title{Marginal log-linear parameters for graphical Markov models}

\author{ {\bf Robin J.~Evans
} \\  Department of Statistics \\  University of Washington\\ \tt{rje42@stat.washington.edu}\\
\and {\bf Thomas S.~Richardson}  \\ Department of Statistics \\ University of Washington\\
\tt{tsr@stat.washington.edu}\\
}

\begin{document}

\maketitle

\begin{abstract}
Marginal log-linear (MLL) models provide a flexible approach to multivariate discrete data. MLL parametrizations under linear constraints induce a wide variety of models, including models defined by conditional independences. 
We introduce a sub-class of MLL models which correspond to Acyclic Directed Mixed Graphs (ADMGs) under the usual global Markov property. We characterize for precisely which graphs the resulting parametrization is variation independent.
The MLL approach provides the first description of ADMG models in terms of a minimal list of constraints. The parametrization is also easily adapted to sparse modelling techniques, which we illustrate using several
examples of real data.
\end{abstract}

\begin{center}
\end{center}

\section{Introduction}

Models defined by conditional independence constraints are central to many methods in multivariate statistics, and in particular to graphical models \citep{darroch:80, whittaker:90}.  In the case of discrete data, \emph{marginal log-linear} (MLL) parameters can be used to parametrize a broad range of models, including some graphical classes and models for conditional independence \citep{rudas:10, forcina:10}.  
These parameters are defined by considering a sequence, $M_1, M_2, \ldots, M_k$, of margins of the distribution which respects inclusion (i.e.\ $M_i$ precedes $M_j$ if $M_i \subset M_j$), with each such sequence giving rise to a smooth parametrization of the saturated model.  Useful sub-models can be induced by setting some of the parameters to zero, or more generally by restricting attention to a linear or affine subset of the parameter space.

The flexibility present in this scheme presents a challenge both in terms of interpreting the resulting model
and performing model selection, for which a tractable search space is typically required. 
We describe a sub-class of marginal log-linear models corresponding to a class of graphs known as \emph{acyclic directed mixed graphs} (ADMGs), which contain directed ($\rightarrow$) and bidirected ($\leftrightarrow$) edges, subject to the constraint that there are no cycles of directed edges; an example is given in Figure \ref{fig:exm}.
The relationship between the MLL models and ADMGs is analogous to that between ordinary log-linear models and undirected graphs: log-linear models give a very rich class of models to choose from, since their number grows doubly-exponentially as the number of variables increases; undirected graphs provide a natural and more manageable subset of models with which to work \citep{darroch:80}.

\begin{figure}
\begin{center}
 \begin{tikzpicture}
 [rv/.style={circle, draw, very thick, minimum size=7mm}, node distance=25mm, >=stealth]
 \pgfsetarrows{latex-latex};
 \node[rv] (1) {1};
 \node[rv, right of=1] (2) {2};
 \node[rv, above of=2, xshift=12.5mm, yshift=-8mm] (3) {3};
 \node[rv, right of=2] (4) {4};
 \draw[->, very thick, color=blue] (1) -- (2);
 \draw[->, very thick, color=blue] (2) -- (4);
 \draw[<->, very thick, color=red] (2) -- (3);
 \draw[<->, very thick, color=red] (3) -- (4);
 \end{tikzpicture}
 \end{center}
\caption{An acyclic directed mixed graph, $\G_1$.}
\label{fig:exm}
\end{figure}
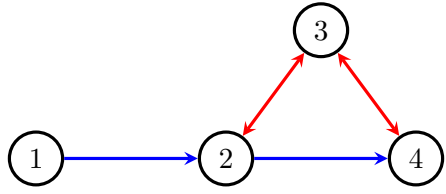

 The patterns of independence described by ADMGs arise naturally in the context of generating processes in which not all variables are observed.  To illustrate this, consider the randomized encouragement design carried out by
  \citet{mcdonald:92} to investigate the effect of computer reminders for doctors on take-up of influenza vaccinations, and consequent morbidity in patients. The study involved 2,861 patients; here we focus on the following fields:
{\it 
\begin{description}
  \item[($\noder$)] patient's doctor sent a card asking to {\textbf{Re}}mind them about flu vaccine (randomized);
  \item[($\nodev$)] patient \textbf{Va}ccinated against influenza;
  \item[($\nodey$)] the endpoint: patient was \emph{not} hospitalized with flu;
  \item[($\nodea$)] \textbf{Ag}e of patient: 0 = `65 and under',  1 = `over 65';
  \item[($\nodec$)] patient has \textbf{C}hronic \textbf{O}bstructive Pulmonary Disease (COPD), as measured at baseline.
\end{description}
}

The graphs in Figure \ref{fig:flu} represent two possible data generating processes. Under both structures, whether
or not a patient's doctor received a reminder note is independent of the baseline variables age ($\nodea$) and COPD status ($\nodec$), as would be expected under randomization. Further the absence of an edge $\noder \rightarrow \nodey$ encodes the assumption that whether or not a reminder ($\noder$) was received only influences the final outcome ($\nodey$) via whether or not a patient received a flu vaccination ($\nodev$). Both structures also assume that there are unobserved confounding factors between
vaccination and COPD, and between COPD and the final outcome. However, the graph in Figure \ref{fig:flu}(b) supposes that there is no additional confounding between $\nodev$ and $\nodey$. As a consequence the generating process given in (b) implies the additional restriction that ${\noder} \indep {\nodey} \mid {\nodev}, {\nodea}$. (We make no assumptions about the state spaces of the variables $\nodeh$, $\nodeh_1$ and $\nodeh_2$, since these factors are unobserved.)

In Figure \ref{fig:flu2} we show the ADMGs corresponding to the generating processes in Figure \ref{fig:flu}. These graphs only contain observed variables, but by including bidirected edges ($\leftrightarrow$) they encode the same observable conditional independence relations; see \S \ref{subsec:globalmarkov} for details.

All the work herein can easily be extended to graphs which also contain an undirected component, provided no undirected edge is adjacent to an arrowhead.  This latter case is equivalent to the summary graphs of \citep{wermuth:11}, and
strictly includes all ancestral graphs \citep{richardson:02}.
Our approach may be seen as extending earlier work \citep{rudas:06,rudas:10,forcina:10} which described the conditional independence structure of certain marginal log-linear models.

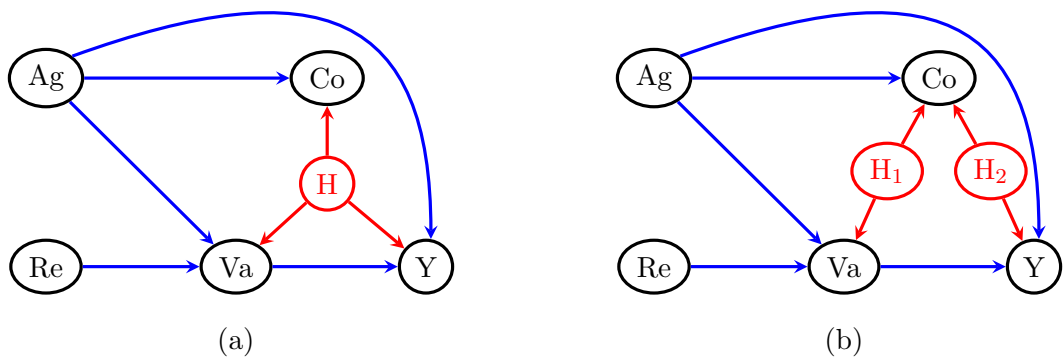
\begin{figure}
\begin{center}
 \begin{tikzpicture}
 [rv/.style={ellipse, draw, very thick, minimum size=7mm, inner sep=1mm}, node distance=25mm, >=stealth]
 \pgfsetarrows{latex-latex};
  \begin{scope}
 \node[rv] (1) {$\noder$};
 \node[rv, right of=1] (2) {$\nodev$};
  \node[below of=2, yshift=15mm] {(a)};
 \node[rv, right of=2] (3) {$\nodey$};
 \node[rv, above of=1, yshift=0mm] (4) {$\nodea$};
 \node[rv, above of=2, xshift=12mm, yshift=0mm] (5) {$\nodec$};
  \node[rv, below of=5,yshift=11mm,color=red] (H) {$\nodeh$};
 \draw[->, very thick, color=blue] (1) -- (2);
 \draw[->, very thick, color=red] (H) -- (2);
 \draw[->, very thick, color=red] (H) -- (5);
  \draw[->, very thick, color=red] (H) -- (3);
 \draw[->, very thick, color=blue] (2) -- (3);
 \draw[->, very thick, color=blue] (4) -- (2);
 \draw[->, very thick, color=blue] (4) -- (5);
\draw[->, very thick, color=blue] (4.40) .. controls +(20:3.5) and +(90:3) .. (3.80);
\end{scope}
 \begin{scope}[xshift=8cm]
 \node[rv] (1) {$\noder$};
 \node[rv, right of=1] (2) {$\nodev$};
  \node[below of=2, yshift=15mm] {(b)};
 \node[rv, right of=2] (3) {$\nodey$};
 \node[rv, above of=1, yshift=0mm] (4) {$\nodea$};
 \node[rv, above of=2, xshift=12.5mm, yshift=0mm] (5) {$\nodec$};
  \node[rv, above right of=2,xshift=-12mm,yshift=-5mm,color=red] (H1) {$\nodeh_1$};
    \node[rv, above left of=3,xshift=12mm, yshift=-5mm,color=red] (H2) {$\nodeh_2$};
 \draw[->, very thick, color=blue] (1) -- (2);
 \draw[->, very thick, color=red] (H1) -- (2);
 \draw[->, very thick, color=red] (H1) -- (5);
  \draw[->, very thick, color=red] (H2) -- (3);
   \draw[->, very thick, color=red] (H2) -- (5);
 \draw[->, very thick, color=blue] (2) -- (3);
 \draw[->, very thick, color=blue] (4) -- (2);
 \draw[->, very thick, color=blue] (4) -- (5);
\draw[->, very thick, color=blue] (4.40) .. controls +(20:3.5) and +(90:3) .. (3.80);
\end{scope}
  \end{tikzpicture}
  \end{center}
    \caption{Two different generating processes for the flu vaccine encouragement design (red vertices are unobserved): both graphs imply ${\noder} \indep {\nodea}, {\nodec}$; however (b) also implies ${\noder} \indep {\nodey} \mid {\nodev}, {\nodea} $.}
    \label{fig:flu}
\end{figure}

\begin{figure}
\begin{center}
 \begin{tikzpicture}
 [rv/.style={ellipse, draw, very thick, minimum size=7mm}, node distance=25mm, >=stealth]
 \pgfsetarrows{latex-latex};
  \begin{scope}
 \node[rv] (1) {$\noder$};
 \node[rv, right of=1] (2) {$\nodev$};
  \node[below of=2, yshift=15mm] {(a)};
 \node[rv, right of=2] (3) {$\nodey$};
 \node[rv, above of=1, yshift=-6mm] (4) {$\nodea$};
 \node[rv, above of=2, xshift=12mm, yshift=-6mm] (5) {$\nodec$};
 \draw[->, very thick, color=blue] (1) -- (2);
  \draw[<->, very thick, color=red] (5) -- (2);
   \draw[<->, very thick, color=red] (5) -- (3);
    \draw[<->, very thick, color=red] (2) to[out=30,in=150] (3);
 \draw[->, very thick, color=blue] (2) -- (3);
 \draw[->, very thick, color=blue] (4) -- (2);
 \draw[->, very thick, color=blue] (4) -- (5);
\draw[->, very thick, color=blue] (4.40) .. controls +(20:3.5) and +(90:3) .. (3.80);
\end{scope}
 \begin{scope}[xshift=8cm]
 \node[rv] (1) {$\noder$};
 \node[rv, right of=1] (2) {$\nodev$};
  \node[below of=2, yshift=15mm] {(b)};
 \node[rv, right of=2] (3) {$\nodey$};
 \node[rv, above of=1, yshift=-6mm] (4) {$\nodea$};
 \node[rv, above of=2, xshift=12.5mm, yshift=-6mm] (5) {$\nodec$};
 \draw[->, very thick, color=blue] (1) -- (2);
 \draw[<->, very thick, color=red] (2) -- (5);
   \draw[<->, very thick, color=red] (3) -- (5);
 \draw[->, very thick, color=blue] (2) -- (3);
 \draw[->, very thick, color=blue] (4) -- (2);
 \draw[->, very thick, color=blue] (4) -- (5);
\draw[->, very thick, color=blue] (4.40) .. controls +(20:3.5) and +(90:3) .. (3.80);
\end{scope}
  \end{tikzpicture}
  \end{center}
    \caption{Two ADMGs representing the conditional independence restrictions on the observed margin implied by the corresponding graphs in Figure \protect{\ref{fig:flu}.}}
    \label{fig:flu2}
\end{figure}
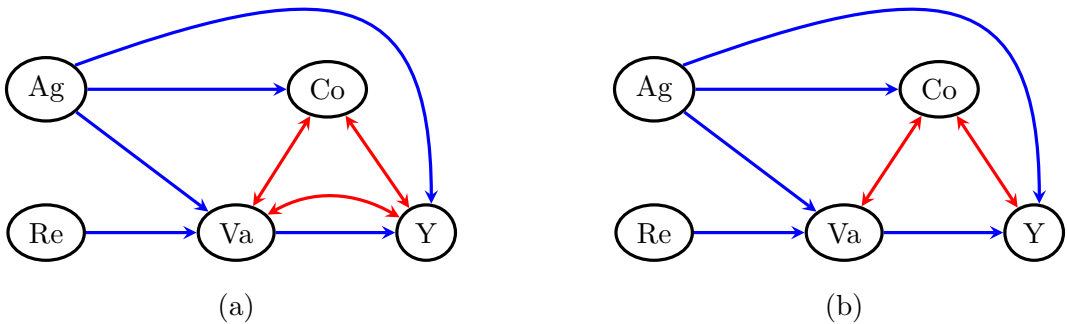

\subsection{ADMG Models}
\citet{richardson:03} described local and global Markov properties for ADMGs, while \citet{richardson:09} described
a parametrization for  discrete random variables via a collection of conditional probabilities of the form
$P(X_H=0\,|\, X_T=x_T)$.  However, although Richardson's parametrization is simple, it does not naturally lead to parsimonious sub-models. In addition, the parameters are subject to variation dependence constraints, in the sense that setting some parameters to particular values may restrict the valid range of other parameters; this makes maximum likelihood fitting, for example, more challenging \citep{evans:10}.  To illustrate this point, consider
 the graph $\G_1$ in Figure \ref{fig:exm} as an example; it encodes the model under which $X_1 \indep X_3$ and $X_4 \indep X_1 \,|\, X_2$.  Richardson's parametrization consists in this case (for binary random variables) of the probabilities
\begin{align*}
&P(X_1 = 0) \qquad \qquad P(X_2 = 0 \;|\; X_1 = x_1) \qquad \qquad P(X_2 = 0, X_3 = 0 \,|\, X_1 = x_1)\\
&P(X_3 = 0) \qquad \qquad 
P(X_4 = 0 \;|\; X_2 = x_2) \qquad \qquad P(X_3 = 0, X_4 = 0 \,|\, X_1=x_1, \, X_2 = x_2)
\end{align*}
where $x_1, x_2 \in \{0,1\}$.  A disadvantage of this parametrization is that, for instance, the joint probabilities $P(X_2 = 0, X_3 = 0 \,|\, X_1 = x_1)$
are bounded above by the marginal probabilities  $P(X_2 = 0 \,|\, X_1 = x_1)$.
Consequently, from the point of view of parameter interpretation, it makes little sense to consider the joint probabilities in isolation.  For example, strong (conditional) correlation between $X_2$ and $X_3$ is present when the joint probability is large relative to the marginals.

However, replacing the joint probabilities $P(X_2 = 0, X_3 = 0 \,|\, X_1 = x_1)$ with the conditional odds ratios
\begin{align*}
& \frac{P(X_2 = 0, X_3 = 0 \;|\; X_1 = x_1) \cdot P(X_2 = 1, X_3 = 1 \;|\; X_1 = x_1)}{P(X_2 = 1, X_3 = 0 \;|\; X_1 = x_1) \cdot P(X_2 = 0, X_3 = 1 \;|\; X_1 = x_1)}, && x_1 \in \{0,1\}
\end{align*}
(and similarly for $P(X_3 = 0, X_4 = 0 \,|\, X_1 = x_1, \, X_2 = x_2)$) yields a variation independent parametrization, the odds ratio measuring dependence without reference to marginal distributions.  
This means that if we wish to define a prior distribution over the univariate probabilities and the odds ratios, we may, if appropriate, simply use a product of univariate distributions; similarly, to fit a generalized linear model with these parameters as joint responses, we need only use simple univariate link functions. 
We will see that this approach to discrete parametrizations can be generalized using marginal log-linear parameters.

In Section \ref{sec:MLLs} we introduce marginal log-linear (MLL) parameters and some of their properties, while Section \ref{sec:ADMGs} gives background theory about ADMGs and the parametrization of \citet{richardson:09}.
The development of MLL parameters for ADMG models is presented in Section \ref{sec:ing}, resulting in a parametrization we refer to as \emph{ingenuous} (since it arises naturally, but `natural parametrization' already has a particular meaning).
We also show that this parametrization can always be embedded in a larger one corresponding to a complete graph and the saturated model, where some of the parameters in the bigger model are linearly constrained.  
In Section \ref{sec:vi} we classify for which models the ingenuous parametrization is variation independent,
since this can facilitate interpretation of the resulting coefficients.
In Section \ref{sec:sparse} we discuss approaches to sparse modelling using MLLs in the context of several additional
datasets and a simulation.
Longer proofs are in Section \ref{sec:proofs}.

\section{Marginal Log-Linear Parameters} \label{sec:MLLs}

We consider collections of random variables $(X_{v})_{ v\in V}$ with finite index set $V$, taking values in finite discrete probability spaces $({\X}_{v})_{v\in V}$ under a strictly positive probability measure $P$; without loss of generality, $\X_v = \{0, 1, \ldots, |\X_v|-1\}$.  For $A\subseteq V$ we let ${\X}_{A} \equiv \times_{v\in A} ({\X}_{v})$, ${\X}\equiv {\X}_{V}$ and similarly $X_{A}\equiv(X_{v})_{v\in A}$, $X\equiv X_{V}$  and $\ii_A \equiv (x_v)_{v \in A}$,
$\ii \equiv x_V$.
%
In addition $\tilde \X$ is the subset of $\X$ which does not contain the last possible element in any co-ordinate; that is $\tilde\X_v = \{0, 1, \ldots, |\X_v|-2\}$, and $\tilde\X = \times_{v \in V} (\tilde\X_v)$.
We use $p_A(\ii_A) \equiv P(X_A = \ii_A)$ and $p_{A|B}(\ii_A \,|\, \ii_B) \equiv P(X_A = \ii_A \,|\, X_B = \ii_B)$, for particular instantiations of $\ii$. 

Following \citet{br:02}, we define a general class of parameters on discrete distributions.  The definition relies upon abstract collections of subsets, so it may be helpful to the reader to keep in mind that the sets $M_i \in \mathbb{M}$ are margins, or subsets, of the distribution over $V$, and each set $\mathbb{L}_i$ is a collection of effects in the margin $M_i$.  A pair $(L, M_i)$ corresponds to a log-linear interaction over the set $L$, within the margin $M_i$.

\begin{definition} \label{definition:BRpre}
For $L \subseteq M \subseteq V$, the pair $(L, M)$ is an ordered pair of subsets of $V$.  Let $\PP$ be a collection of such pairs, and define
\begin{align*}
\mathbb{M} \equiv \{M \;|\; (L, M) \in \PP \mbox{ for some }L\},
\end{align*}
to be the collection of margins in $\PP$.
If $\mathbb{M} = \{M_1, \ldots, M_k\}$, write
\begin{align*}
\mathbb{L}_i \equiv \{L \;|\; (L, M_i) \in \PP\},
\end{align*}
for the set of effects present in the margin $M_i$.
We say that the collection $\PP$ is \emph{hierarchical} if the ordering on $\mathbb{M}$ may be chosen so that if $i < j$, then $M_j \nsubseteq M_i$ and also
$L \in \mathbb{L}_j \Rightarrow L \nsubseteq M_i$;
the second condition is equivalent to saying that each $L$ is associated only with the first margin $M$ of which it is a subset.  We say the collection is \emph{complete} if every non-empty subset of $V$ is an element of precisely one set $\mathbb{L}_i$.
\end{definition}

The term `hierarchical' is used because each log-linear interaction is defined in the first possible margin in an ascending class, and `complete' because all interactions are present.  Some authors \citep{rudas:10, lupparelli:09} consider only collections which are complete. 

\begin{definition}
For each $M \subseteq V$ and $\ii_M \in \X_M$, define the functions $\lambda_L^M(\ii_L)$ by the identity
\begin{align*}
\log p_M(\ii_M) \equiv \sum_{L \subseteq M} \lambda_L^M(\ii_L),
\end{align*}
subject to the identifiability constraint that for every $\emptyset \neq L \subseteq M$, $\ii_L \in \X_L$ and $v \in L$,
\begin{align*}
\sum_{x_v \in \X_v} \lambda_L^M(\ii_{L\setminus\{v\}}, x_v) = 0;
\end{align*}
that is, the sum over the support of each variable is zero.  We call $\lambda_L^M(\ii_L)$ a \emph{marginal log-linear parameter}.

Note that the constant $\lambda_{\emptyset}^M$ is determined by the values of the other parameters and the fact that the probabilities $p_M(\ii_M)$ sum to one.  In the sequel we will always assume that $L$ is non-empty.
\end{definition}

The term `marginal log-linear parameter' is coined by analogy with ordinary log-linear parameters, which correspond to  the special case $M=V$.  The following result provides an explicit expression for $\lambda_L^M(\ii_L)$.

\begin{lemma} \label{lem:mllp}
For $L \subseteq M \subseteq V$ and $\ii_L \in \X_L$ we have
\begin{align}
\lambda_L^M(\ii_L) &= \frac{1}{|\X_{M}|} \sum_{\jj_M \in \X_M}  \log p_M(\jj_M) \prod_{v \in L} \left( |\X_{v}| \mathbb{I}_{\{x_v = y_v\}} - 1 \right). \label{eqn:mllp}
\end{align}
\end{lemma}
\noindent This result is elementary, and its proof is omitted. 
\bigskip

For a collection of ordered pairs of subsets $\PP$ (see Definition \ref{definition:BRpre}), we let
\begin{align*}
\tilde{\Lambda}(\PP) &= \{\lambda_L^M(\ii_L) \;|\; (L, M) \in \PP, \ii_L \in \tilde{\X}_L \}
\end{align*}
be the collection of marginal log-linear parameters associated with $\PP$. 
Note that we avoid the redundancy created by the identifiability constraint by only considering $\ii_L \in \tilde\X_L$.


The definition of a marginal log-linear parameter we give is equivalent to the recursive one given in \citet{br:02}; since both expositions are somewhat abstract, we invite the reader to consult the examples below for additional intuition.  In particular note that for binary random variables, the product in (\ref{eqn:mllp}) is always $\pm 1$.  \citet[Theorem 2]{br:02} show that any collection $\tilde\Lambda(\PP)$ where $\PP$ is hierarchical and complete smoothly parametrizes the saturated model, that is, it parametrizes the set of all positive distributions on $\X$.

The restriction that the parameters must sum to zero is required for identifiability, but different constraints can be used in its place.
We might instead require that $\lambda_L^M(\ii_L)$ be zero whenever any entry of $\ii_L$ is zero (or some other selected value); this is seen in \citet{marchetti:11}, for example, and its use would not substantially affect any of the results in this paper.
%

\subsection{Examples of Marginal Log-Linear Models}

We will write $\lambda_L^M$ to mean the collection $\{\lambda_L^M(\ii_L) \;|\; \ii_L \in \X_L\}$; the expression $\lambda_L^M = 0$ denotes that we are setting all the parameters in this collection to 0.

\begin{example} 
The classical log-linear parameters for a discrete distribution over a set of variables $V$ are $\{\lambda_L^V \;|\; L \subseteq V \}$.  
\end{example}

\begin{example} 
Up to trivial transformations, the multivariate logistic parameters of \citet{glonek:95} are $\{\lambda_L^L \;|\; L \subseteq V \}$.
\end{example}

\begin{example}
Let $V = \{1,2,3\}$ and assume all random variables are binary.  Write $P_{001} \equiv P(X_1 = 0, X_2 = 0, X_3 = 1)$, and $P_{1++} \equiv P(X_1 = 1)$, etc.   Then
\begin{align*}
\lambda_1^1(0) = \frac{1}{2} \log \frac{P_{0 ++}}{P_{1 ++}},
\end{align*}
which, up to a multiplicative constant, is the logit of the probability of the event $\{X_1 = 0\}$.  Also,
\begin{align*}
\lambda_{1}^{12}(0) = \frac{1}{4} \log \frac{P_{00+} \; P_{01+}}{P_{10+} \; P_{11+}}\qquad \mbox{and}\qquad\lambda_{12}^{12}(0,0) = \frac{1}{4} \log \frac{P_{00+} \; P_{11+}}{P_{10+} \; P_{01+}},
\end{align*}
the log odds product and log odds ratio between $X_1$ and $X_2$ respectively.  

If instead $X_1$ is ternary, we obtain
\begin{align*}
\lambda_1^1(0) = \frac{1}{3} \log \frac{P_{0 ++}^2}{P_{1 ++} \; P_{2 ++}},
\end{align*}
\begin{align*}
\lambda_{1}^{12}(0) = \frac{1}{6} \log \frac{P_{00+}^2 \; P_{01+}^2}{P_{10+} \; P_{11+} \; P_{20+} \; P_{21+}}
\qquad \mbox{and}\qquad\lambda_{12}^{12}(0,0) = \frac{1}{6} \log \frac{P_{00+}^2 \; P_{11+}  \; P_{21+}}{P_{10+} \; P_{20+} \; P_{01+}^2}.
\end{align*}
Here $\lambda_1^1(0)$ contrasts the probability $P(X_1 = 0)$ with the geometric mean of the probabilities $P(X_1 = 1)$ and $P(X_1 = 2)$.  On the other hand, up to constants, $\lambda_{12}^{12}(0,0)$ is an average of the two log odds ratios
\begin{align*}
& \log \frac{P_{00+} \; P_{21+}}{P_{20+} \; P_{01+}} && \log \frac{P_{00+} \; P_{11+}}{P_{10+} \; P_{01+}},
\end{align*}
and so gives a contrast between $P(X_1 = X_2 = 0)$ and other joint probabilities in a way which generalizes the binary log odds ratio and provides a measure of dependence; in particular note that $\lambda_{12}^{12}(0,0) = 0$ if $X_1  \indep X_2$.

Here we have written, for example, $12$ instead of $\{1,2\}$; similarly, for sets $A$ and $B$ we sometimes write $AB$ for $A \cup B$, and $aB$ for $\{a\} \cup B$.
\end{example}

%


\subsection{Properties of Marginal Log-Linear Models}

The next result relates marginal log-linear parameters to conditional independences; it is found as Lemma 1 in \citet{rudas:10} and Equation (6) of \citet{forcina:10}.

\begin{lemma} \label{lem:brindep}
For any disjoint sets $A$, $B$ and $C$, where $C$ may be empty, $A \indep B \;|\; C$ if and only if
\begin{align*}
\lambda_{A'B'C'}^{ABC} = 0 \qquad \mbox{for every} \quad\emptyset \neq A' \subseteq A, \quad \emptyset \neq B' \subseteq B, \quad C' \subseteq C.
\end{align*}
\end{lemma}

The special case of $C = \emptyset$ (giving marginal independence) is proved in the context of multivariate logistic parameters by \citet{kauermann:97}.

\begin{example}
Take a complete and hierarchical parametrization of 3 variables,
\begin{align*}
\lambda_1^1 \qquad \lambda_2^2 \qquad \lambda_3^3 \qquad \lambda_{12}^{12} \qquad \lambda_{13}^{13} \qquad \lambda_{23}^{123} \qquad \lambda_{123}^{123}.
\end{align*}
Then we can force $X_1 \indep X_3$ by setting $\lambda_{13}^{13} = 0$.  Similarly $X_2 \indep X_3 \;|\; X_1$ corresponds to setting $\lambda_{23}^{123} = \lambda_{123}^{123} = 0$.
\end{example}

The following lemma shows that under conditional independence constraints, certain MLL parameters defined within different margins are equal.

\begin{lemma} \label{lem:rbnnat}
Suppose that $A \indep B \; | \; C$, and $A$ is non-empty.  Then for any $D \subseteq C$,
\begin{align*}
\lambda_{AD}^{ABC}(\ii_{AD}) = \lambda_{AD}^{AC}(\ii_{AD}), \qquad \mbox{for each } \ii_{AD} \in \X_{AD}.
\end{align*}
\end{lemma}

The proof of this result is found in Section \ref{subsec:lem:rbnnat}.

\section{Acyclic Directed Mixed Graphs} \label{sec:ADMGs}

We introduce basic graphical concepts used to describe the global Markov property and parametrization schemes.

\begin{definition}
A \emph{directed mixed graph} $\G$ consists of a set of vertices $V$, and both directed ($\rightarrow$) and bidirected ($\leftrightarrow$) edges.  Edges of the same type and orientation may not be repeated, but there may be multiple edges of different types between a pair of vertices.

A \emph{path} in $\G$ is a sequence of adjacent edges, without repetition of a vertex; a path may be empty, or equivalently consist of only one vertex.  The first and last vertices on a path are the \emph{endpoints} (these are not distinct if the path is empty); other vertices on the path are \emph{non-endpoints}.
The graph $\G_1$ in Figure \ref{fig:exm}, for example, contains the path $1 \rightarrow 2 \rightarrow 4 \leftrightarrow 3$, with endpoints 1 and 3, and non-endpoints 2 and 4.  
A \emph{directed} path is one in which all the edges are directed ($\rightarrow$) and are oriented in the same direction, whereas a \emph{bidirected path} consists entirely of bidirected edges.  

A directed cycle is a non-empty sequence of edges of the form $v \rightarrow \cdots \rightarrow v$.  An \emph{acyclic} directed mixed graph (ADMG) is one which contains no directed cycles.
\end{definition}

\begin{definition}
For a graph $\G$ and a subset of its vertices $A \subseteq V$, we denote by $\G_A$ the \emph{induced subgraph} formed by $A$; that is, the graph containing the vertices $A$, and the edges in $\G$ whose endpoints are both in $A$.
\end{definition}

\begin{definition}
Let $a$ and $d$ be vertices in a mixed graph $\G$.  If $a=d$, or there is a directed path from $a$ to $d$, we say that $a$ is an \emph{ancestor} of $d$, and that $d$ is a \emph{descendant} of $a$.  The sets of ancestors of $d$ and descendants of $a$ are denoted $\an_{\G}(d)$ and $\dec_{\G}(a)$ respectively.
If there is a directed path from $a$ to $d$ containing precisely one edge ($a \rightarrow d$) then $a$ is called a \emph{parent} of $d$; the set of vertices which are parents of $d$ is written $\pa_{\G} (d)$.

The \emph{district} of $a$, denoted $\dis_{\G}(a)$, is the set containing $a$ and all vertices which are connected to $a$ by a bidirected path.  
These definitions are applied disjunctively to sets of vertices, so that, for example,
\[
\pa_{\cal G}(W)\equiv \bigcup_{w\in W}\pa_{\cal G}(w),\quad\quad
\dis_{\cal G}(W)\equiv \bigcup_{w\in W}\dis_{\cal G}(w).
\]
A set of vertices $A$ is \emph{ancestral} if $A = \an_{\G} (A)$; that is, $A$ contains all its own ancestors.
\end{definition}

\begin{example}
Consider the graph $\G_1$ in Figure \ref{fig:exm}.  We have
\begin{align*}
\an_{\G_1}(4) = \{1,2,4\} \qquad \an_{\G_1}(\{2,3\}) = \{1,2,3\}.
\end{align*}
The district of 3 is the set $\{2,3,4\}$, and since 3 has no parents, $\pa_{\G_1}(3) = \emptyset$.
\end{example}

Note that by the definitions of some authors, vertices are not their own ancestors \citep{lau:96}.  The above notations may be shortened on induced subgraphs so that $\pa_A \equiv \pa_{\G_A}$, and similarly for other definitions.  In some cases where the meaning is clear, we will dispense with the subscript altogether.

We use the now standard notation of \citet{dawid:condind}, and represent the statement `$X$ is independent of $Y$ given $Z$ under a probability measure $P$', for random variables $X$, $Y$ and $Z$, by $X \indep Y \;|\; Z \; [P]$.  If $P$ is unambiguous, this part is dropped, and if $Z$ is empty we write simply $X \indep Y$.  Finally, we abuse notation in the usual way: $v$ and $X_v$ are used interchangeably as both a vertex and a random variable; likewise $A$ denotes both a vertex set and $X_{A}$.

\subsection{Global Markov Property for ADMGs}\label{subsec:globalmarkov}

A Markov property associates 
a particular set of independence relations with a graph.

A non-endpoint vertex $c$ on a path is a \emph{collider} on the path if the edges preceding and succeeding $c$ on the path have an arrowhead at $c$, for example $\rightarrow  c \leftarrow$ or $\leftrightarrow  c\leftarrow$; otherwise $c$ is a \emph{non-collider}.  A path between vertices $a$ and $b$ in a mixed graph is said to be \emph{blocked given a set} $C$ if either
\begin{itemize}
\item[(i)] there is a non-collider on the path in $C$,  or
\item[(ii)] there is a collider on the path which is \emph{not} in $\an_\G(C)$.
\end{itemize}
If all paths from $a$ to $b$ are blocked by $C$, then $a$  and $b$ are said to be {\it m-separated given} $C$.
Sets $A$ and $B$ are said to be m-separated given $C$ if every $a \in A$ and every $b \in B$ are  m-separated given $C$.  This naturally extends the d-separation criterion of \citet{pearl:88} to graphs with bidirected edges.

A probability measure $P$ on $\mathfrak X$ is said to satisfy the {\it global Markov property} for $\cal G$ if for every triple of disjoint sets of vertices $A$, $B$ and $C$,
\begin{center}
$A$ is m-separated from $B$ given $C$ in $\G \qquad \Longrightarrow \qquad X_A \indep X_B  \mid X_C \;[P]$.
\end{center}
The {\em model} associated with an ADMG $\G$ is simply the set of distributions that obey the global Markov property for
$\G$.

\begin{proposition}\label{prop:m-connecting-path}
 If a path m-connects $x$ and $y$ given $Z$ in $\G$ then every vertex on the path is in $\an_{\G}(\{x,y\}\cup Z)$.
\end{proposition}

\begin{proof} This follows from the definition of m-connection.\end{proof}

\begin{example}
Consider the graph $\G_1$ in Figure \ref{fig:exm}.  There are two paths between the vertices 1 and 4,
\begin{align*}
\pi_1 \,:\, 1 \rightarrow 2 \rightarrow 4 \qquad \mbox{and} \qquad \pi_2 \,:\, 1 \rightarrow 2 \leftrightarrow 3 \leftrightarrow 4;
\end{align*}
both are blocked by $C=\{2\}$.  $\pi_1$ is blocked because 2 is a non-collider on the path and is in $C$, while $\pi_2$ is blocked because 3 is a collider on the path and is not in $\an_{\G_1}(C) = \{1,2\}$.  Hence $\{1\}$ and $\{4\}$ are m-separated given $\{2\}$ in $\G_1$.

One can similarly see that $\{1\}$ and $\{3\}$ are m-separated given $C=\emptyset$, and that no other m-separations hold for this graph.  Thus a joint distribution $P$ obeys the global Markov property for $\G_1$ if and only if $X_1 \indep X_4 \,|\, X_2 \, [P]$ and $X_1 \indep X_3 \, [P]$. 
\end{example}
By similar arguments the independences associated with the ADMGs in Figure \ref{fig:flu} may also be read off.

\subsection{Existing Parametrization of ADMG models} \label{subsec:exists}

This subsection defines the parameters of \citet{richardson:09} for multivariate discrete distributions satisfying the global Markov property for an ADMG.

\begin{definition}
Let $\G$ be an ADMG with vertex set $V$.  We say that a collection of vertices $W \subseteq V$ is \emph{barren} if for each $v \in W$, we have $W \cap \dec_{\G} (v) = \{v\}$; in other words $v$ has no non-trivial descendants in $W$.  For an arbitrary set of vertices $U$, the maximal subset with no non-trivial descendants in $U$ is denoted $\barren_{\G}(U)$.

A \emph{head} is a collection of vertices $H$ which is connected by bidirected paths in $\G_{\an (H)}$ and is barren in $\G$.  We write $\mathcal{H}(\G)$ for the collection of heads in $\G$.
The \emph{tail} of a head $H$ is the set
\begin{align*}
\tail_{\G} (H) \equiv \pa_{\G}(\dis_{\an(H)} (H)) \cup (\dis_{\an(H)} (H) \setminus H).
\end{align*}
Thus the tail of $H$ is the set of vertices in $V\setminus H$ connected to a vertex in $H$
by a path on which every vertex is a collider and an ancestor of a vertex in $H$.
We typically write $T$ for a tail, provided it is clear which head it belongs to.
\end{definition}

\begin{proposition}\label{prop:head-and-tail}
Let $H$ be a head.  Then {\rm (i)} $H = \barren_{\G}(H \cup \tail_{\G} (H))$; {\rm (ii)} $\tail_{\G}(H) \subseteq \an_{\G}(H)$.
\end{proposition}

\begin{proof} Immediate from the respective definitions.
\end{proof}

\citet{richardson:09} shows that discrete distributions obeying the global Markov property for an ADMG $\G$ are
parametrized by the conditional probabilities:
\begin{align*}
\left\{P(X_H = \ii_H \;|\; X_T = \ii_T) \;\middle|\;H \in \mathcal{H}, \; T = \tail_{\G} (H), \; \ii_H \in \tilde{\X}_H, \; \ii_T \in \X_T \right\}.
\end{align*}
This is achieved via factorizations based on head-tail pairs; let $\prec$ be the partial ordering on heads such that $H_i \prec H_j$ if $H_i \subset \an_\G(H_j)$ and $H_i \neq H_j$.  This is well defined, since otherwise $\G$ would contain a directed cycle.  Then let
$[\cdot]_{\G}$ be a function which partitions sets of vertices into heads by repeatedly removing heads which are maximal under $\prec$.

Then $P$ satisfies the global Markov property for $\G$ if and only if it obeys the factorizations
\begin{align}
P(X_A = \ii_A) &= \prod_{H\in [A]_{\cal G}} P(X_H = \ii_H \;|\; X_T = \ii_T) \label{eqn:pform}
\end{align}
for ancestral sets of vertices $A$;  see \citet{richardson:09} for details.
In the case of a directed acyclic graph (DAG), this corresponds to the probability distribution of each vertex conditional on its parents.

\begin{example}
Consider again the ADMG $\G_1$ in Figure \ref{fig:exm}; its head-tail pairs $(H,T)$ are $(1, \emptyset)$, $(2, 1)$, $(3, \emptyset)$, $(23, 1)$, $(4, 2)$ and $(34, 12)$.
Multivariate binary distributions obeying the global Markov property with respect to $\G_1$ are therefore parametrized by
\begin{align*}
p_1(0) \qquad p_{2|1}(0 \,|\, x_1) \qquad p_{3}(0) \qquad p_{23|1}(0,0,\,|\, x_1)\\
p_{4|2}(0 \,|\, x_2) \qquad p_{34|12}(0,0 \,|\, x_1, x_2),
\end{align*}
for $x_1, x_2 \in \{0,1\}$, as mentioned in the Introduction.
\end{example}

\subsection{Graphical Completions}

Given a discrete model defined by a set of conditional independence constraints, it is natural to consider it as a sub-model of the saturated model, which contains all positive probability distributions.  In a setting where the model is graphical, it becomes equally natural to think of the graph as a subgraph of a complete graph, by which we mean a graph containing at least one edge between every pair of vertices.  
We can obtain a complete graph from an incomplete one by inserting edges between each pair of vertices which lack one, but this leaves a choice of edge type and orientation.  
These choices may affect how much of the structure and spirit of the original graph is retained; we will require that a completion preserves the heads of the original graph, which helps to preserve the structure of the parametrization.

\begin{definition}
Given an ADMG $\G$ and a supergraph $\bar{\G}$, we call $\bar{\G}$ a \emph{head-preserving completion} of $\G$ if $\bar{\G}$ is complete, and $\mathcal{H}(\G) \subseteq \mathcal{H}(\bar\G)$.
\end{definition}

It is easy to see that a head-preserving completion always exists; for example, if we add in a bidirected edge between every pair of vertices which are not joined by an edge, then it is clear that barren sets in $\G$ will remain barren in $\bar\G$, and bidirected connected sets in $\G$ will remain bidirected connected in $\bar\G$.

Note that it is not necessary for every pair of vertices to be joined by an edge in order for a graph to represent the saturated model, however we will require this for our completions.

\begin{example}
Figure \ref{fig:comp} shows a head-preserving completion of the ADMG in Figure \ref{fig:exm}.

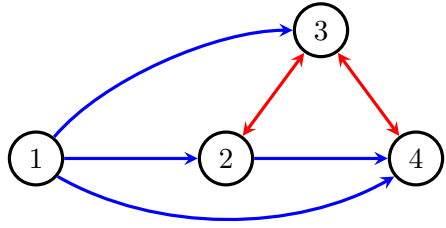
\begin{figure}
\begin{center}
 \begin{tikzpicture}
 [rv/.style={circle, draw, very thick, minimum size=7mm}, node distance=25mm, >=stealth]
 \pgfsetarrows{latex-latex};
 \node[rv] (1) {1};
 \node[rv, right of=1] (2) {2};
 \node[rv, above of=2, xshift=12.5mm, yshift=-8mm] (3) {3};
 \node[rv, right of=2] (4) {4};
 \draw[->, very thick, color=blue] (1) -- (2);
 \draw[->, very thick, color=blue] (2) -- (4);
 \draw[<->, very thick, color=red] (2) -- (3);
 \draw[<->, very thick, color=red] (3) -- (4);
 \draw[<-, very thick, color=blue] (3.180) .. controls +(180:1) and +(50:1) .. (1.50);
 \draw[->, very thick, color=blue] (1.320) .. controls +(330:1.5) and +(210:1.5) .. (4.220);
  \end{tikzpicture}
  \end{center}
\caption{A head-preserving completion, $\bar{\G}_1$, of the ADMG in Figure \ref{fig:exm}.}
\label{fig:comp}
\end{figure}
\end{example}

\begin{proposition}\label{prop:ancestors-in-completion}
If $\bar{\G}$ is a head-preserving completion of $\G$ then $\an_{\G}(v) \subseteq \an_{\bar{\G}}(v)$. In particular,
if a set $A$ is ancestral in $\bar{\G}$ then $A$ is also ancestral in $\G$.
\end{proposition}

\begin{proof} This follows because $\G$ contains a subset of the edges in $\bar{\G}$.\end{proof}



\section{Ingenuous Parametrization of an ADMG model} \label{sec:ing}

We now use the marginal log-linear parameters defined in Section \ref{sec:MLLs} to parametrize the ADMG models discussed in Section \ref{sec:ADMGs}.

\begin{definition}
Consider an ADMG $\G$ with head-tail pairs $(H_i, T_i)$ over some index $i$, and let $M_i = H_i \cup T_i$.  Further, let $\mathbb{L}_i = \{A \; | \; H_i \subseteq A \subseteq H_i \cup T_i\}$.  This collection of margins and associated effects is the \emph{ingenuous} parametrization of $\G$, denoted $\PP^{\ing}(\G)$.
\end{definition}

\begin{example}
We return again to the ADMG $\G_1$ in Figure \ref{fig:exm}; the head-tail pairs are $(1, \emptyset)$, $(2,1)$, $(3, \emptyset)$, $(23, 1)$, $(4, 2)$ and $(34, 12)$, meaning that the ingenuous parametrization is given by the following margins and effects:

\begin{center}
\begin{tabular}{r|l}
$M$&$\mathbb{L}$\\
\hline
1&1\\
12&2, 12\\
3&3\\
123&23, 123\\
24&4, 24\\
1234&34, 134, 234, 1234.\\
\end{tabular}
\end{center}
\end{example}

Note that the ordering of the margins given here is hierarchical; in order to use most of the results of \citet{br:02}, we need to confirm that the definition above always leads to a hierarchical parametrization, which is shown by the following result.

\begin{lemma} \label{lem:hier}
For any ADMG $\G$, there is an ordering on the margins $M_i$ of the ingenuous parametrization $\PP^{\ing}(\G)$ which is hierarchical.
\end{lemma}

\begin{proof}
Firstly we show that for distinct heads $H_i$ and $H_j$, the collections $\mathbb{L}_i$ and $\mathbb{L}_j$ are disjoint.  To see this, assume for a contradiction that there exists $A$ such that $H_i \subseteq A \subseteq H_i \cup T_i$ and $H_j \subseteq A \subseteq H_j \cup T_j$.  Since $H_i \neq H_j$, assume without loss of generality that there exists $v \in H_i \cap H_j^c \subseteq A$.  

Then $v \in H_j \cup T_j$ implies that $v \in T_j$, and thus there is a directed path from $v$ to some $w \in H_j$.  Now, $w \notin H_i$, since $v,w \in H_i$ would imply that $H_i$ is not barren.  But if $w \in H_j \cap H_i^c$, then by the same argument as above we can find a directed path from $w$ to some $x \in H_i$.  Then $v \rightarrow \cdots \rightarrow w \rightarrow \cdots \rightarrow x$ is a directed path between elements of $H_i$, which is a contradiction.  Thus $\mathbb{L}_i$ and $\mathbb{L}_j$ are disjoint.

Now, 
consider the partial ordering $\prec$ of heads defined in Section \ref{subsec:exists}: $H_i \prec H_j$ whenever $H_i \subset \an_{\G}(H_j)$ and $H_i \neq H_j$.  
Any total ordering which respects this partial ordering is hierarchical, because each set $A \in \mathbb{L}_i$ is a subset of the ancestors of $H_i$.
\end{proof}

We proceed to show that the ingenuous parameters for an ADMG $\G$ characterize the set of distributions which obey the global Markov property with respect to $\G$.

\begin{lemma} \label{lem:param}
For any sets $M$ and $L \subseteq M$, the collection of MLL parameters
\begin{align*}
\{\lambda_A^M(\ii_A) \;|\; L \subseteq A \subseteq M, \ii_M \in \tilde\X_M \},
\end{align*}
together with the $(|L|-1)$-dimensional marginal distributions of $X_L$ conditional on $X_{M \setminus L}$, smoothly parametrizes the distribution of $X_L$ conditional on $X_{M \setminus L}$.
\end{lemma}

A proof is given in Section \ref{subsec:lem:param}.

We now come to the main result of this section.

\begin{theorem} \label{thm:param}
The ingenuous parametrization $\tilde\Lambda(\PP^{\ing}(\G))$ of an ADMG $\G$ parametrizes precisely those distributions $P$ obeying the global Markov property with respect to $\G$.
\end{theorem}

\begin{proof}  We proceed by induction.
Again we use the partial ordering $\prec$ on heads from Section \ref{subsec:exists}. 
For the base case, we know that singleton heads $\{h\}$ with empty tails are parametrized by the logits $\lambda_h^h$.

Now, suppose that we wish to find the distribution of a head $H$ conditional on its tail $T$.  Assume that we have the distribution of all heads $H'$ which precede $H$, conditional on their respective tails; we claim this is sufficient to give the $(|H|-1)$-dimensional marginal distributions of $H$ conditional on $T$.  

Let $v \in H$, and let $C = H \setminus \{v\}$ be a $(|H|-1)$-dimensional marginal of interest.  The set $A = \an_{\G}(H) \setminus \{v\}$ is ancestral, since $v$ cannot have (non-trivial) descendants in $\an_{\G}(H)$; in particular $C \cup T \subseteq A$.  Theorem 4 of \citet{richardson:09} states that the factorization in equation (\ref{eqn:pform}) holds for every ancestral set, so
\begin{align*}
p_A(\ii_A) &= \prod_{\substack{H' \in [A]_{\G} \\ T' = \tail (H)}} p_{H'|T'}(\ii_{H'} \,|\, \ii_{T'}).
\end{align*}
But all the probabilities in the product are known by our induction hypothesis, and the marginal distribution of $C$ conditional on $T$ is given by the distribution of $A$.

The ingenuous parametrization, by definition, contains $\lambda_A^{H\cup T}$ for $H \subseteq A \subseteq H\cup T$, and thus the result follows from Lemma \ref{lem:param}.
\end{proof}

\begin{example} \label{exm:small}
Returning to our running example, the graph $\G_1$ in Figure \ref{fig:exm} corresponds to the model
\begin{align*}
\Big\{ P \,\Big|\, X_1 \indep X_4 \,|\, X_2 \, [P] \mbox{ and } X_1 \indep X_3 \, [P]\Big\}.
\end{align*}
Theorem \ref{thm:param} tells us that this collection of distributions is precisely characterized by the ingenuous parameters for $\G_1$,
\begin{align*}
&\lambda_1^1 \qquad\qquad \lambda_2^{12} \qquad \lambda_{12}^{12} \qquad\qquad \lambda_3^{3} \qquad\qquad \lambda_{23}^{123} \qquad \lambda_{123}^{123}\\
&\lambda_4^{24} \qquad \lambda_{24}^{24} \qquad\qquad \lambda_{34}^{1234} \qquad \lambda_{134}^{1234} \qquad \lambda_{234}^{1234} \qquad \lambda_{1234}^{1234}.
\end{align*}
\end{example}

\subsection{Constraint-Based Model Description}

The results above show that the ingenuous parameters for an ADMG $\G$, like Richardson's parameters, provide precisely the information required to reconstruct a distribution obeying the global Markov property for $\G$.
However, it is difficult to use this parametrization in practice unless we can evaluate the likelihood, which requires us to make explicit the map which we have implicitly defined from the ingenuous parameters to the joint probability distribution under the model.
For example, for the parameters in \citet{richardson:09} there is an explicit map from the parameters back to the joint distribution using a generalization of M\"{o}bius inversion. This was used by \citet{evans:10} to fit these models via maximum likelihood.
In contrast, the map from ingenuous parameters to the joint distribution cannot be written in closed form.

An alternative approach is to consider the ingenuous parametrization as part of a larger, complete parametrization of the saturated model, such that the additional parameters are constrained to be zero under the sub-model defined by $\G$.  
This enables us to fit the model using Lagrange-type algorithms, as in \citet{evans:11}. 

\begin{theorem} \label{thm:subsp}
Let $\G$ be an ADMG, and $\bar{\G}$ a head-preserving completion of $\G$.
The ingenuous parametrization of $\G$ corresponds to setting
\begin{align*}
\lambda_L^M = 0
\end{align*}
for $(L,M) \in \mathbb{P}^{\ing}(\bar{\G})$ whenever $L$ does not appear as an effect in $\mathbb{P}^{\ing}(\G)$.  In particular, these constraints define the set of distributions which satisfy the global Markov property with respect to $\G$.
\end{theorem}

The proof of this result is found in Section \ref{subsec:thm:subsp}

\begin{example}
Consider again the ADMG $\G_1$ in Figure \ref{fig:exm}; a possible head-preserving completion $\bar{\G}_1$ (shown in Figure \ref{fig:comp}) is obtained by adding the edges $1 \rightarrow 3$ and $1 \rightarrow 4$.
The ingenuous parametrization for $\bar{\G}_1$ is
\begin{center}
\begin{tabular}{r|l}
$M$&$\mathbb{L}$\\
\hline
1&1\\
2&2, 12\\
13&3, 13\\
123&23, 123\\
124&4, 14, 24, 124\\
1234&34, 134, 234, 1234.\\
\end{tabular}
\end{center}

The effects found in $\PP^{\ing}(\bar{\G}_1)$ but not in $\PP^{\ing}(\G_1)$ are 13, 14, and 124, and indeed the sub-model defined by $\G_1$ corresponds to setting 
\begin{align*}
\lambda_{13}^{13} = \lambda_{14}^{124} = \lambda_{124}^{124} = 0;
\end{align*}
under this model the following equalities hold by Lemma \ref{lem:rbnnat}:
\begin{align*}
&\lambda_{4}^{124} = \lambda_{4}^{24} && \lambda_{24}^{124} = \lambda_{24}^{24}.
\end{align*}
Removing the zero parameters in $\PP^{\ing}(\bar\G_1)$ and renaming two others according to the above equations returns us to the ingenuous parametrization of $\G_1$.

Theorem \ref{thm:subsp} shows that we can fit the model defined by $\G_1$ by maximum likelihood simply by maximizing the log-likelihood subject to $\lambda_{13}^{13} = \lambda_{14}^{124} = \lambda_{124}^{124} = 0$.
In particular, this approach always provides a list of independent constraints which characterize the model.
\end{example}

An obvious question which arises is whether \emph{any} completion of a graph will lead to a complete parametrization with the property of Theorem \ref{thm:subsp}.
We can obtain a counterexample by considering the complete graph $\tilde{\G}_1$ in Figure \ref{fig:nosub}, which has ingenuous parametrization
\begin{center}
\begin{tabular}{r|l}
$M$&$\mathbb{L}$\\
\hline
3&3\\
13&1, 13\\
123&2, 12, 23, 123\\
1234&4, 14, 24, 124, 34, 134, 234, 1234.\\
\end{tabular}
\end{center}
The graph $\G_1$ in Figure \ref{fig:exm} is a subgraph of $\tilde{\G}_1$, and corresponds to the model 
obtained by setting $\lambda_{13}^{13} = \lambda_{14}^{124} = \lambda_{124}^{124} = 0$; however, these last two parameters do not appear in the ingenuous parametrization of $\tilde{\G}_1$, and so there is no way to enforce the sub-model as a linear constraint.  

$\tilde{\G}_1$ is, of course, not head-preserving.  Such completions may still lead to parametrizations which satisfy the property of Theorem \ref{thm:subsp}: for example, if the edge $1 \rightarrow 3$ is added to the graph in Figure \ref{fig:chains}(a), this destroys the head $\{1,2,3\}$, but the sub-model corresponds to $\lambda_{13}^{13} = 0$, which is a parameter in the complete graph.
%

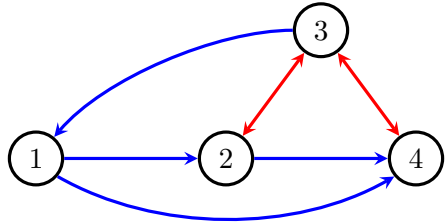
\begin{figure}
\begin{center}
 \begin{tikzpicture}
 [rv/.style={circle, draw, very thick, minimum size=7mm}, node distance=25mm, >=stealth]
 \pgfsetarrows{latex-latex};
 \node[rv] (1) {1};
 \node[rv, right of=1] (2) {2};
 \node[rv, above of=2, xshift=12.5mm, yshift=-8mm] (3) {3};
 \node[rv, right of=2] (4) {4};
 \draw[->, very thick, color=blue] (1) -- (2);
 \draw[->, very thick, color=blue] (2) -- (4);
 \draw[->, very thick, color=blue] (3.180) .. controls +(180:1) and +(50:1) .. (1.50);
 \draw[<->, very thick, color=red] (2) -- (3);
 \draw[<->, very thick, color=red] (3) -- (4);
 \draw[->, very thick, color=blue] (1.320) .. controls +(330:1.5) and +(210:1.5) .. (4.220);
 \end{tikzpicture}
 \end{center}
\caption{A complete ADMG, $\tilde{\G}_1$, of which $\G_1$ is a subgraph, but whose ingenuous parametrization does not contain the model described by $\G_1$ as a linear sub-space because the associated completion is not head-preserving.}
\label{fig:nosub}
\end{figure}

\subsection{Relationship To Prior Work}

\citet{rudas:10} parametrize chain graph models of multivariate regression type, also known as type IV chain graph models, using marginal log-linear parameters.  Type IV chain graph models are a special case of ADMG models, in the sense that by replacing the undirected edges in a type IV chain graph with bidirected edges, the global Markov property on the resulting ADMG is equivalent to the Markov property for the chain graph \citep[see][]{drton:09}.  
The graphs in Figure \ref{fig:chains} are examples of Type IV models.  
However, there are models in the class of ADMGs which do not correspond to any chain graph, such as the one described by $\G_1$ in Figure \ref{fig:exm}.

The parametrization of \citet{rudas:10} uses different choices of margins to the ingenuous parametrization, though their parameters can be shown to be equal to the parameters considered here under the global Markov property, using Lemma \ref{lem:rbnnat}.  Thus the variation dependence properties of that parametrization are identical to those of the ingenuous parametrization (see next section).  \citet{forcina:10} provide an algorithm which gives a range of `admissible' margins in which collections of conditional independence constraints may be defined.

\citet{marchetti:11} also parametrize type IV chain graph models in a similar manner to \citet{rudas:10}, in that case using multivariate logistic contrasts.

\section{Variation Independence} \label{sec:vi}

As discussed in the introduction, the interpretation of parameters and the construction of prior distributions
is simpler when parameters are variation independent.

\begin{definition}
Let $\theta_i$, for $i = 1,\ldots,k$ be a collection of parameters such that $\theta_i$ takes all values in the set $\Theta_i$.  We say that the vector $\boldsymbol \theta = (\theta_1, \ldots, \theta_k)$ is \emph{variation independent} if $\boldsymbol \theta$ can take every value in the set $\Theta_1 \times \cdots \times \Theta_k$.
\end{definition}

\citet{br:02} characterize precisely which hierarchical and complete para\-metrizations are variation independent, using a notion they call ordered decomposability.
We now do this for ingenuous parametrizations.  

\begin{definition}
A collection of sets $\mathbb{M} = \{M_1, \ldots, M_k\}$ is \emph{incomparable} if $M_i \nsubseteq M_j$ for every $i \neq j$.

A collection $\mathbb{M}$ of incomparable subsets of $V$ is \emph{decomposable} if it has at most two elements, or there is an ordering $M_1, \ldots, M_k$ on the elements of $\mathbb{M}$ wherein for each $i = 3, \ldots, k$, there exists $j_i < i$ such that
\begin{align*}
\left(\bigcup_{l=1}^{i-1} M_l\right) \cap M_i = M_{j_i} \cap M_i.
\end{align*}
This is also known as the \emph{running intersection property}.

A collection $\mathbb{M}$ of (possibly comparable) subsets is \emph{ordered decomposable} if it has at most two elements, or there is an ordering $M_1, \ldots, M_k$ such that $M_i \nsubseteq M_j$ for $i>j$,
and for each $i = 3,\ldots, k$, the inclusion maximal elements of $\{M_1, \ldots, M_i\}$ form a decomposable collection.  We say that a collection $\PP$ of parameters is ordered decomposable if there is an ordering on the margins $\mathbb{M}$ which is both hierarchical and ordered decomposable.
\end{definition}

The following example is found in \citet{br:02}.

\begin{example} \label{exm:nori}
Let 
$\mathbb{M} = \{12,13,23,123\}$.  In order to have a hierarchical ordering of these margins it is clear that the set $123$ must come last, but there is no way to order the collection of inclusion maximal margins $\{12,13,23\}$ such that it has the running intersection property.  Thus $\mathbb{M}$ is not ordered decomposable.
\end{example}

The next result links variation independence to ordered decomposability.

\begin{theorem}[\citet{br:02}, Theorem 4]
Let $\mathbb{P}$ be a parametrization which is hierarchical and complete.  Then the parameters $\tilde\Lambda(\PP)$ are variation independent if and only if $\mathbb{P}$ is ordered decomposable.
\end{theorem}

As previously noted, the ingenuous parametrization is not complete in general, and so we cannot apply the above result directly to characterize its variation dependence.  
However, by constructing complete parametrizations of which the ingenuous parametrizations are linear sub-models, we can obtain the following.


\begin{theorem} \label{thm:ingVI}
The ingenuous parametrization for an ADMG $\G$ is variation independent if and only if $\G$ contains no heads of size greater than or equal to 3.
\end{theorem}

The proof of this result is found in Section \ref{subsec:thm:ingVI}.

\begin{example}
The graph $\G_1$ in Figure \ref{fig:exm} has maximum head size 2, and therefore the associated ingenuous parametrization is variation independent.

Likewise the graphs in Figure \ref{fig:flu2}(a) and (b) contain no heads of size greater than $2$, so that the resulting ingenuous parameters are variation independent.  Note that this was not true of the parameters given by \citet{richardson:09}.
\end{example}


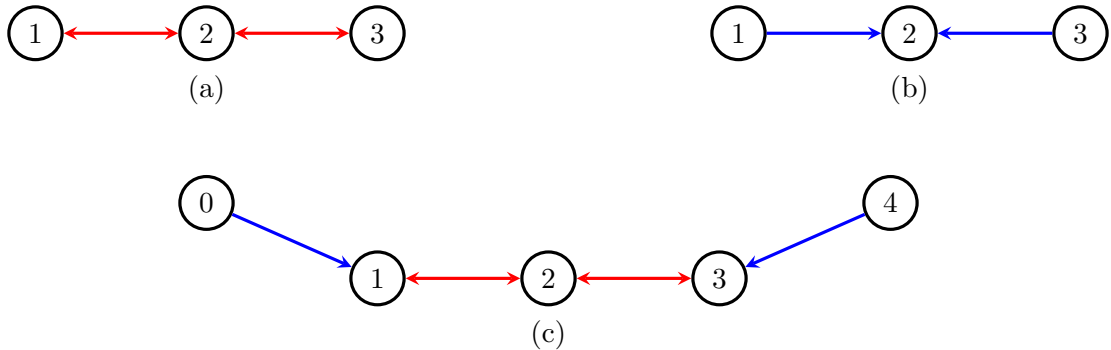
\begin{figure}
\begin{center}
 \begin{tikzpicture}
 [rv/.style={circle, draw, very thick, minimum size=7mm}, node distance=22.5mm, >=stealth]
 \pgfsetarrows{latex-latex};
 \node[rv] (1) {1};
 \node[rv, right of=1] (2) {2};
 \node[rv, right of=2] (3) {3};
 \draw[<->, very thick, color=red] (1) -- (2);
 \draw[<->, very thick, color=red] (2) -- (3);
 \node[below of=2, yshift=15mm] {(a)};
 \node[rv, right of=3, xshift=25mm] (1a) {1};
 \node[rv, right of=1a] (2a) {2};
 \node[rv, right of=2a] (3a) {3};
 \draw[->, very thick, color=blue] (1a) -- (2a);
 \draw[<-, very thick, color=blue] (2a) -- (3a);
 \node[below of=2a, yshift=15mm] {(b)};
 \node[rv, below of=2, yshift=0mm] (0b) {0};
 \node[rv, right of=0b, yshift=-10mm] (1b) {1};
 \node[rv, right of=1b] (2b) {2};
 \node[rv, right of=2b] (3b) {3};
 \node[rv, right of=3b, yshift=10mm] (4b) {4};
 \draw[->, very thick, color=blue] (0b) -- (1b);
 \draw[->, very thick, color=blue] (4b) -- (3b);
 \draw[<->, very thick, color=red] (1b) -- (2b);
 \draw[<->, very thick, color=red] (2b) -- (3b);
 \node[below of=2b, yshift=15mm] {(c)};
  \end{tikzpicture}
  \end{center}
\caption{(a) a graph with a variation dependent ingenuous parametrization; (b) a Markov equivalent graph to (a) with a variation independent ingenuous parametrization; (c) a graph with no variation independent MLL parametrization.}
\label{fig:chains} 
\end{figure}

\begin{example}
The bidirected 3-chain shown in Figure \ref{fig:chains}(a) has the head $123$, and therefore its ingenuous parametrization is variation dependent.  This can easily be seen directly: in the binary case, for example, if the parameters $\lambda_{12}^{12}(0)$ and $\lambda_{23}^{23}(0)$ are chosen to be very large, this induces very strong dependence between the variables $X_1$ and $X_2$, and between $X_2$ and $X_3$ respectively.  If these correlations are chosen to be too large, then it is impossible for $X_1$ and $X_3$ to be marginally independent, which is implied by the graph.

Observe that we could use the Markov equivalent graph in Figure \ref{fig:chains}(b), which has no heads of size 3, and thus obtain a variation independent parametrization of the same model.  However, if we add incident arrows as shown in Figure \ref{fig:chains}(c), we obtain a graph where such a trick is not possible.  
In fact this third graph has no variation independent parametrization in the Bergsma and Rudas framework, since it requires $\lambda_{0124}^{0124} = \lambda_{0134}^{0134} = \lambda_{0234}^{0234} = 0$, and these margins cannot be ordered in a way which satisfies the running intersection property (see Example \ref{exm:nori}).
\end{example}

In general, it would be sensible for a statistician concerned about variation dependence to choose a graph from the Markov equivalence class created by their model which has the smallest possible maximum head size.  This could be achieved by reducing the number of bidirected edges in the graph, where possible; see, for example, \citet{ali:05} and \citet{drton:richardson:08a} for algorithms for finding the graph with the minimal number of arrowheads
in a given Markov equivalence class.

\begin{example} \label{exm:lupparelli}

The bidirected 4-cycle, shown in Figure \ref{fig:4cyc}, contains a head of size 4, and so its ingenuous parametrization is variation dependent.  
\begin{figure}
\begin{center}
 \begin{tikzpicture}
 [rv/.style={circle, draw, very thick, minimum size=7mm}, node distance=25mm, >=stealth]
 \pgfsetarrows{latex-latex};
 \node[rv] (1) {1};
 \node[rv, right of=1] (2) {2};
 \node[rv, below of=2] (3) {3};
 \node[rv, below of=1] (4) {4};
 \draw[<->, very thick, color=red] (1) -- (2);
 \draw[<->, very thick, color=red] (2) -- (3);
 \draw[<->, very thick, color=red] (3) -- (4);
 \draw[<->, very thick, color=red] (4) -- (1);
  \end{tikzpicture}
  \end{center}
\caption{A bidirected 4-cycle.}
\label{fig:4cyc} 
\end{figure}
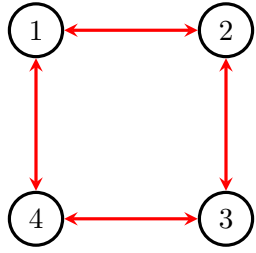
However, there is a marginal log-linear parametrization of this model which is ordered decomposable, and therefore variation independent.
The 4-cycle is precisely the model with $X_1 \indep X_3$ and $X_2 \indep X_4$.  Set $\mathbb{M} = \{13, 24, 1234\}$, with 
\begin{align*}
  \mathbb{L}_1 &= \{1, 3, 13\}\\
  \mathbb{L}_2 &= \{2, 4, 24\}\\
  \mathbb{L}_3 &= \mathscr{P}(\{1,2,3,4\}) \setminus (\mathbb{L}_1 \cup \mathbb{L}_2);
\end{align*}
here $\mathscr{P}(A)$ denotes the power set of $A$.  This gives a hierarchical, complete and ordered decomposable parametrization, so the parameters are variation independent.  The 4-cycle corresponds exactly to setting $\lambda_{13}^{13} = \lambda_{24}^{24} = 0$, and it follows that the remaining parameters are still variation independent under this constraint.
\end{example}

This approach to parametrization, which considers disconnected sets, is discussed in detail by \citet{lupparelli:09}.  It produces a variation independent parametrization for graphs where the disconnected sets do not overlap, and may well be preferable to the ingenuous parametrization in these cases.  In sparser graphs however, it does not seem as useful; as mentioned above, some graphs have no variation independent MLL parametrization.

\section{Parsimonious Modelling with Marginal Log-Linear Parameters} \label{sec:sparse}

The number of parameters in a model associated with a sparse graph containing bidirected edges can, in certain cases, be relatively large. 
In a purely bidirected graph, the parameter count depends upon the number of connected sets of vertices; in the case of a chain of bidirected edges such as that shown in Figure \ref{fig:chain}(a), this means that the number of parameters grows quadratically in the length of the chain.

The parametrization of \citet{richardson:09}, and its special case for purely bidirected graphs \citep[see][]{drton:richardson:08} does not present us with any obvious method of reducing the parameter count whilst preserving the conditional independence structure.  
In contrast, there are well established methods for sparse modelling with other classes of graphical models.  In the case of an undirected graph with binary random variables, restricting to one parameter for each vertex and each edge leads to a Boltzmann Machine \citep{ackley:85}.  \citet{rudas:06} use marginal log-linear parameters to provide a sparse parametrization of a DAG model, again restricting to one parameter for each vertex and edge.

As we will see from the following examples, the ingenuous parametrization allows us to fit graphical models with a large number of parameters, and then remove higher-order interactions to obtain a more parsimonious model whilst preserving the conditional independence structure of the original graph.

\subsection{Flu Vaccination Data Revisited}

We first return to the 
\citet{mcdonald:92} study considered in the Introduction. All variables are binary, and (excepting Age) are coded as 0 = false, $1$ = true; we add constraints to our model sequentially, recording the results in the analysis of deviance Table \ref{table:1}.  The ADMG in Figure \ref{fig:flu2}(a) represents the constraint $\mbox{Ag}, \mbox{Co} \indep \mbox{Re}$; it fits well, having a deviance of 2.54 on 3 degrees of freedom.  The smaller model for \ref{fig:flu2}(b) encodes
\begin{align*}
&\mbox{Ag}, \mbox{Co} \indep \mbox{Re} && \mbox{Y} \indep \mbox{Re} \,|\, \mbox{Va}, \mbox{Ag};
\end{align*}
note that these precise independences cannot be represented by a DAG or chain graph (of any of the types considered by \citet{drton:09}).
It also fits well (deviance 7.66 on 7 d.f.), so we may prefer it on the grounds of simplicity.  


The ingenuous parametrization in this case contains some higher order effects, including the 5-way interaction between all variables.
Setting $\lambda_{L}^M=0$ for $|L| \geq 4$
removes five parameters whilst increasing the deviance by only 2.22; removing the effects of size 3 adds a further 8.39 to the deviance whilst removing seven more parameters.  
The resulting model has a total deviance of 18.28 on 19 degrees of freedom, representing a good fit compared to the saturated model (likelihood ratio test $p=0.49$).

\begin{table}
\begin{center}
\begin{tabular}{|l|c|r|r|r|}
\hline
\textbf{Constraint}&\textbf{Figure}&\textbf{Add.\ Dev.}&\textbf{d.f.}&\textbf{Total Dev.}\\
\hline
$\mbox{Ag}, \mbox{Co} \indep \mbox{Re}$&\ref{fig:flu2}(a)&2.54&3&2.54\\
$\mbox{Y} \indep \mbox{Re} \,|\, \mbox{Va}, \mbox{Ag}$&\ref{fig:flu2}(b)&5.11&7&7.66\\
no 4- and 5-way params&&2.22&12&9.88\\
no 3-way params&&8.39&19&18.28\\
\hline
\end{tabular}
\caption{Analysis of deviance table of models considered for influenza data.  Constraints are added sequentially from top to bottom; the last three columns give the additional deviance for the constraint, the total degrees of freedom and the total deviance of the models respectively.}
\label{table:1}
\end{center}
\end{table}

\subsection{Incorporating Symmetry: Twins Data}

\citet{hakim:03} investigate genetic effects on the presence or absence of two soft tissue disorders, frozen shoulder and tennis elbow, based on a study in pairs of monozygotic and dizygotic twins; the data are reproduced in \citet{ekholm:12}.
We have count data for a 5-way contingency table over the variables
$S_i$ and $E_i$, indicators of whether twin $i$ in the pair suffers from frozen shoulder and tennis elbow respectively, $i \in \{1,2\}$, and 
$T$, an indicator of whether the pair are monozygotic or dizygotic twins.
There are a total of 866 observations for monozygotic pairs, and 963 for dizygotic pairs; twin 1 corresponds to the twin who was born first.

We first fitted the model $T \indep (S_1, S_2, E_1, E_2)$ to test whether the zygosity of the twins has any effect on the other variables; we obtained a deviance of 16.4 on 15 degrees of freedom, suggesting that there is no evidence that $T$ is related to the other variables.  Note that this contradicts the conclusions of \citet{ekholm:12}, but they use additional assumptions to obtain more powerful tests.

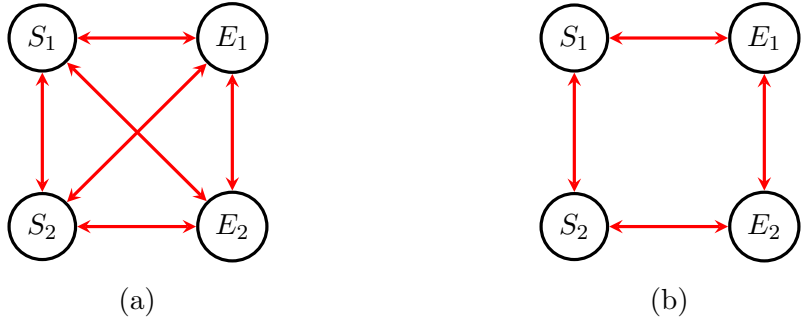
\begin{figure}
\begin{center}
 \begin{tikzpicture}
 [rv/.style={circle, draw, very thick, minimum size=7mm}, node distance=25mm, >=stealth]
 \pgfsetarrows{latex-latex};
 \node[rv] (A1) {$S_1$};
 \node[rv, right of=A1] (P1) {$E_1$};
 \node[rv, below of=A1] (A2) {$S_2$};
 \node[rv, right of=A2] (P2) {$E_2$};
 \draw[<->, very thick, color=red] (A1) -- (P1);
 \draw[<->, very thick, color=red] (A2) -- (P2);
 \draw[<->, very thick, color=red] (A1) -- (A2);
 \draw[<->, very thick, color=red] (A1) -- (P2);
 \draw[<->, very thick, color=red] (P1) -- (P2);
 \draw[<->, very thick, color=red] (P1) -- (A2);
\node[below of=A2, xshift=12.5mm, yshift=15mm] {(a)};
 \node[rv, right of=P1, xshift=20mm] (A1a) {$S_1$};
 \node[rv, right of=A1a] (P1a) {$E_1$};
 \node[rv, below of=A1a] (A2a) {$S_2$};
 \node[rv, right of=A2a] (P2a) {$E_2$};
 \draw[<->, very thick, color=red] (A1a) -- (P1a);
 \draw[<->, very thick, color=red] (A2a) -- (P2a);
 \draw[<->, very thick, color=red] (A1a) -- (A2a);
 \draw[<->, very thick, color=red] (P1a) -- (P2a);
\node[below of=A2a, xshift=12.5mm, yshift=15mm] {(b)};
  \end{tikzpicture}
  \end{center}
    \caption{Graphs for the twins data for  models corresponding to (a) a common gene and (b) separate genes affecting the prevalence of frozen shoulder and tennis elbow.}
    \label{fig:twins}
\end{figure}

Collapsing to a 4-way table over $(S_1, S_2, E_1, E_2)$, we consider the complete bidirected model in Figure \ref{fig:twins}(a).  A further simplifying assumption is to impose symmetry between the twins in each pair, on the basis that we do not expect any association between the prevalence of the disorders and which twin was born first.  Using the ingenuous parametrization for the graph in Figure \ref{fig:twins}(a), which is itself symmetric with respect to the individual twins, this amounts to six independent linear constraints,
and gives a deviance of 0.59 compared to the saturated model on four variables; there is therefore no evidence to reject symmetry.

Now, a hypothesis of interest is whether a common gene is responsible for the increased risk of the two disorders, or the genetic effects are separate and independent.  In the latter case we would expect the data to be explained by the model encoded by the graph in Figure \ref{fig:twins}(b), and therefore to observe the marginal independences $E_1  \indep S_2$ and $E_2 \indep S_1$ \citep[see][for more details]{drton:richardson:08}.  This amounts to the constraint $\lambda_{E_1 S_2}^{E_1 S_2} = \lambda_{E_2 S_1}^{E_2 S_1} = 0$; the first equality already holds by symmetry, so only one additional constraint is imposed.

This model has a deviance of 8.41 on 7 degrees of freedom, which is not rejected in a likelihood ratio test with the saturated model ($p = 0.30$), and so there is no evidence to reject the separate genes hypothesis.  We remark however, that the model with symmetry but no marginal independences has a slightly lower BIC score, and so might be preferred.

The elimination of the 4-way and 3-way interaction parameters for the model from Figure \ref{fig:twins}(b) with symmetry results in deviances of 11.63 on 8 d.f.\ and 16.69 on 10 d.f.\ respectively, both of which also represent reasonable fits; the latter of these has just 5 free parameters.

\subsection{Netherlands Kinship Data}

The Netherlands Kinship Panel Survey (NKPS) is an ongoing study which collects longitudinal information on several thousand Dutch individuals and their families \citep{NKPS:05, NKPS:07}.  
One question asked of both the primary respondents (\emph{anchors}) and their partners is ``How is your health in general?'', with possible responses of `excellent', `good', `good nor poor', `poor' and `very poor'.  We combined `good nor poor', `poor' and `very poor'  into one category to avoid small counts.

Two waves of data are currently available, from 2002--04 and 2006--07. 
We only considered anchors who had the same partner in both waves, and such that both the individual and the partner answered the health question in both waves.
 Let $A_i$ and $P_i$ denote the response of the anchor and partner respectively for wave $i \in \{1,2\}$.
In total there are $n=2,318$ data points, classified into a $3 \times 3 \times 3 \times 3$ table.  

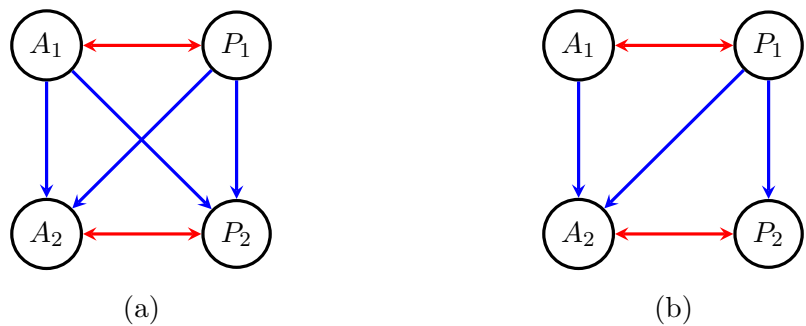
\begin{figure}
\begin{center}
 \begin{tikzpicture}
 [rv/.style={circle, draw, very thick, minimum size=7mm}, node distance=25mm, >=stealth]
 \pgfsetarrows{latex-latex};
 \node[rv] (A1) {$A_1$};
 \node[rv, right of=A1] (P1) {$P_1$};
 \node[rv, below of=A1] (A2) {$A_2$};
 \node[rv, right of=A2] (P2) {$P_2$};
 \draw[<->, very thick, color=red] (A1) -- (P1);
 \draw[<->, very thick, color=red] (A2) -- (P2);
 \draw[->, very thick, color=blue] (A1) -- (A2);
 \draw[->, very thick, color=blue] (A1) -- (P2);
 \draw[->, very thick, color=blue] (P1) -- (P2);
 \draw[->, very thick, color=blue] (P1) -- (A2);
\node[below of=A2, xshift=12.5mm, yshift=15mm] {(a)};
 \node[rv, right of=P1, xshift=20mm] (A1a) {$A_1$};
 \node[rv, right of=A1a] (P1a) {$P_1$};
 \node[rv, below of=A1a] (A2a) {$A_2$};
 \node[rv, right of=A2a] (P2a) {$P_2$};
 \draw[<->, very thick, color=red] (A1a) -- (P1a);
 \draw[<->, very thick, color=red] (A2a) -- (P2a);
 \draw[->, very thick, color=blue] (A1a) -- (A2a);
 \draw[->, very thick, color=blue] (P1a) -- (P2a);
 \draw[->, very thick, color=blue] (P1a) -- (A2a);
\node[below of=A2a, xshift=12.5mm, yshift=15mm] {(b)};
  \end{tikzpicture}
  \end{center}
    \caption{Graphs for the NKPS data; responses of \textbf{A}nchor and \textbf{P}artner regarding their assessment of health; subscripts indicate time. (a) a complete graph; (b) a subgraph which implies $P_2 \indep A_1 \,|\, P_1$.}
    \label{fig:nkps}
\end{figure}

We begin with the complete graph in Figure \ref{fig:nkps}.  One plausible model would be that anchors and their partners are exchangeable.  Since the graph is symmetrical in this respect, so is the ingenuous parametrization, and enforcing symmetry amounts merely to a set of 36 linear constraints; for example:
\begin{align*}
\lambda_{A_2 P_2}^{A_1 P_1 A_2 P_2}(1, 0) &= \lambda_{A_2 P_2}^{A_1 P_1 A_2 P_2}(0, 1).
\end{align*}
This model has a deviance of 89.98, which when compared to the tail of a $\chi^2_{36}$ distribution gives $p = 1.6 \times 10^{-6}$; thus the symmetry model is a poor fit to the data, and is rejected.
The lack of exchangeability is probably due to selection bias in the sampling of the anchors, as well as the different ways in which the anchors and their partners were asked the question: anchors were asked about their health as part of a face-to-face interview, whereas the partners were only asked to complete a survey.  See \citet{siemiatycki:79} for an analysis of differences resulting from survey mode.



If instead we remove the edge $A_1 \rightarrow P_2$ and fit the graph in Figure \ref{fig:nkps}(b), we obtain an explanation of the data which is not rejected at the 5\% level (deviance 19.09 on 12 degrees of freedom, $p = 0.086$); this model corresponds to the conditional independence $P_2 \indep A_1 \,|\, P_1$.  
This graph is the only subgraph of the complete graph in Figure \ref{fig:nkps}(a) which leads to a good fit; in particular the model created by removing the edge $P_1 \rightarrow A_2$ is strongly rejected, which is one manifestation of the asymmetry between individuals and their partners.

Note that we could also have obtained the independence $P_2 \indep A_1 \,|\, P_1$, for instance, by using a DAG with topological ordering $P_1$, $A_1$, $P_2$, $A_2$, but the resulting parametrization would have made it much more difficult to enforce the symmetry constraint tested above. 

\subsection{Example: Trust Data}

\citet{drton:richardson:08} examine responses to seven questions relating to trust and social institutions, taken from the US General Social Survey between 1975 and 1994.  Briefly, the seven questions were:
\begin{description}
	\item[Trust.] Can most people be trusted?
	\item[Helpful.] Do you think most people are usually helpful?
	\item[MemUn, MemCh.] Are you a member of a labour union / church?
	\item[ConLegis, ConClerg, ConBus.] Do you have confidence in congress / organized religion / business?
\end{description}
In that paper, the model given by the graph in Figure \ref{fig:trust} is shown to adequately explain the data, having a deviance of 32.67 on 26 degrees of freedom, when compared with the saturated model.  The authors also provide an undirected graphical model which has one more edge than the graph in Figure \ref{fig:trust}, and yet has 62 fewer parameters.  It too gives a good fit to the data, having a deviance of 87.62 on 88 degrees of freedom.  Both graphs were chosen by backwards stepwise selection methods; see \citet{drton:richardson:08} for details.

\begin{figure}
\begin{center}
 \begin{tikzpicture}
 [rv/.style={ellipse, draw, very thick, inner sep=1mm, minimum width=15mm}, node distance=35mm, >=stealth]
 \pgfsetarrows{latex-latex};
 \node[rv] (1) {Trust};
 \node[rv, above of=1, yshift=-15mm] (2) {Helpful};
 \node[rv, left of=1] (4) {MemCh};
 \node[rv, below of=4, yshift=15mm] (3) {MemUn};
 \node[rv, right of=1] (6) {ConClerg};
 \node[rv, right of=3] (7) {ConBus};
 \node[rv, right of=7] (5) {ConLegis};
 \draw[<->, very thick, color=red] (1) -- (2);
 \draw[<->, very thick, color=red] (1) -- (4);
 \draw[<->, very thick, color=red] (2) -- (4);
 \draw[<->, very thick, color=red] (4.10) .. controls +(20:2) and +(160:2) .. (6.170);
 \draw[<->, very thick, color=red] (2.310) .. controls +(310:2) and +(50:2) .. (7.50);
 \draw[<->, very thick, color=red] (1) -- (6);
 \draw[<->, very thick, color=red] (1) -- (7);
 \draw[<->, very thick, color=red] (2) -- (6);
 \draw[<->, very thick, color=red] (6) -- (7);
 \draw[<->, very thick, color=red] (4) -- (7);
 \draw[<->, very thick, color=red] (3) -- (4);
 \draw[<->, very thick, color=red] (3) -- (7);
 \draw[<->, very thick, color=red] (5) -- (6);
 \draw[<->, very thick, color=red] (5) -- (7);
  \end{tikzpicture}
  \end{center}
\caption{Markov model for trust data given in \citet{drton:richardson:08}.}
    \label{fig:trust}
\end{figure}
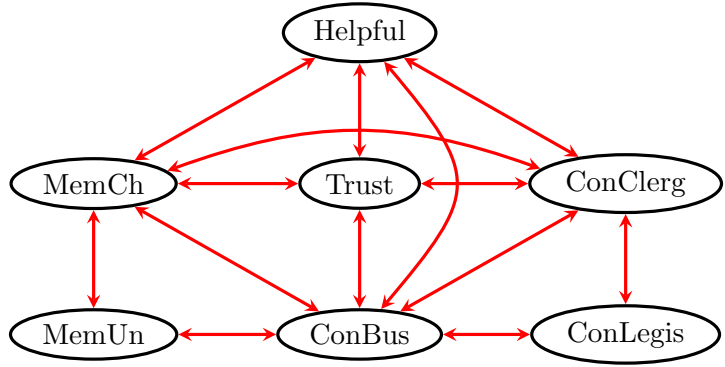

For practical and theoretical reasons, the bidirected model may be preferred to the undirected one, even though the latter appears to be much more parsimonious.  
One may consider the dependence between the responses given to a questionnaire to be manifestations of unmeasured characteristics of the respondent, such as their political beliefs.  
Such a system can be well represented by a bidirected graph, through its marginal independence structure and connection to latent variable models, but not necessarily by an undirected one, which induces conditional independences.  
Note that, since models defined by undirected and bidirected graphs are not nested, there is no \emph{a priori} reason to expect the two methods to give a similar graphical structure.

The greater parsimony of the undirected model (when defined purely by conditional indep\-endences) is due to its hierarchical nature: if we remove an edge between two vertices $a$ and $b$, then this corresponds to requiring that $\lambda_A^V = 0$ for every effect $A$ containing both $a$ and $b$.  
Removing that edge in a bidirected model may correspond merely to setting $\lambda_{ab}^{ab} = 0$ and nothing else, depending upon the other edges present.  Using the ingenuous parametrization, it is easy to constrain additional higher order terms to be zero to obtain sub-models of the set of distributions obeying the global Markov property.

Starting with the model in Figure \ref{fig:trust} and fixing the 4-, 5-, 6- and 7-way interaction terms to be zero increases the deviance to 84.18 on 81 degrees of freedom; none of the 4-way interaction parameters was found to be significant on its own.  Furthermore, removing 21 of the remaining 25 three-way interaction terms increases the deviance to 111.48 on 102 degrees of freedom; using an asymptotic $\chi^2$ approximation gives a p-value of 0.245, so this model is not contradicted by the data.  The only parameters retained are the one-dimensional marginal probabilities, the two-way interactions corresponding to edges in Figure \ref{fig:trust}, and the following three-way interactions:
\begin{align*}
&\mbox{MemUn}, \mbox{ConClerg}, \mbox{ConBus} && \mbox{Helpful}, \mbox{MemUn}, \mbox{MemCh}\\[4pt]
&\mbox{Trust}, \mbox{ConLegis}, \mbox{ConBus} && \mbox{MemCh}, \mbox{ConClerg}, \mbox{ConBus}.
\end{align*}
This model retains the marginal independence structure of Drton and Richardson's model, but provides a good fit with only 25 parameters, rather than the original 101.

A similar analysis, for different data, is performed by \citet[][page 573]{lupparelli:09}; again they find an undirected graphical model to be much more parsimonious than any bidirected one, but obtain comparable fits by removing statistically insignificant higher-order parameters.


\subsection{Simulated Data}

\begin{figure}
\begin{center}
 \begin{tikzpicture}
 [rv/.style={circle, draw, very thick, minimum size=7mm}, node distance=25mm, >=stealth]
 \pgfsetarrows{latex-latex};
 \node[rv] (1) {1};
 \node[rv, right of=1] (2) {2};
 \node[rv, right of=2] (3) {3};
 \node[right of=3] (4) {$\cdots$};
 \node[rv, right of=4] (5) {$k$};
 \draw[<->, very thick, color=red] (1) -- (2);
 \draw[<->, very thick, color=red] (2) -- (3);
 \draw[<->, very thick, color=red] (3) -- (4);
 \draw[<->, very thick, color=red] (4) -- (5);
 \node[below of=3, yshift=15mm] {(a)};
 \node[rv, below of=1, yshift=-10mm] (1a) {1};
 \node[rv, right of=1a] (2a) {2};
 \node[rv, right of=2a] (3a) {3};
 \node[right of=3a] (4a) {$\cdots$};
 \node[rv, right of=4a] (5a) {$k$};
 \node[rv, color=gray, xshift=12.5mm, yshift=10mm] (h1) at (1a) {$h_1$};
 \node[rv, color=gray, xshift=12.5mm, yshift=10mm] (h2) at (2a) {$h_2$};
 \node[rv, color=gray, xshift=12.5mm, yshift=10mm] (h3) at (3a) {$h_3$};
 \node[rv, inner sep=0.1mm, color=gray, xshift=12.5mm, yshift=10mm] (h4) at (4a) {$h_{k-1}$};
 \draw[->, very thick, color=blue] (h1) -- (1a);
 \draw[->, very thick, color=blue] (h1) -- (2a);
 \draw[->, very thick, color=blue] (h2) -- (2a);
 \draw[->, very thick, color=blue] (h2) -- (3a);
 \draw[->, very thick, color=blue] (h3) -- (3a);
 \draw[->, very thick, color=blue] (h3) -- (4a);
 \draw[->, very thick, color=blue] (h4) -- (4a);
 \draw[->, very thick, color=blue] (h4) -- (5a);
 \node[below of=3a, yshift=15mm] {(b)};
  \end{tikzpicture}
  \end{center}
    \caption{(a) A bidirected $k$-chain and (b) a DAG with latent variables ($h_1, \ldots, h_{k-1}$) generating the same observable conditional independence structure.}
    \label{fig:chain}
\end{figure}
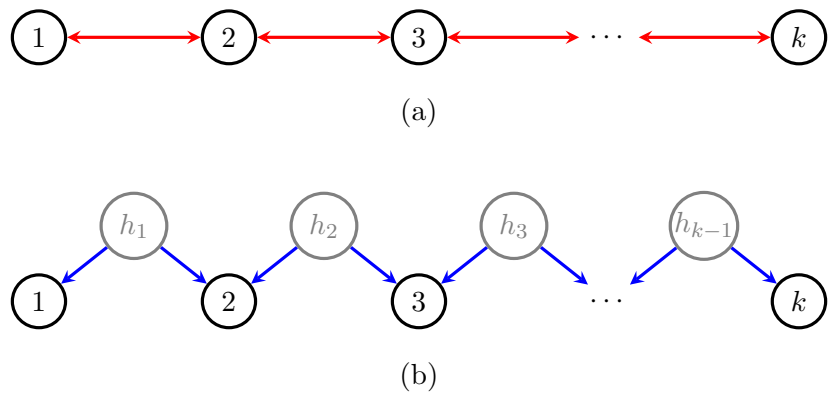

We saw in the earlier examples that we were often able to remove higher order interaction parameters without  compromising
the goodness of fit. Here we explore this phenomena further via simulations.

Consider the DAG with latent variables shown in Figure \ref{fig:chain}(b); over the observed variables, the  conditional independences which hold are exactly those given by the bidirected chain in Figure \ref{fig:chain}(a).

We randomly generated 1,000 distributions from this DAG model with $k=6$, where each latent variable was given three states, and each observed variable two.  The probability of each observed variable being zero, conditional on each state of its parents, was an independent uniform random draw on $(0,1)$; latent states were fixed to occur with equal probability.  For each distribution, a sample size of 10,000 was drawn, and the bidirected chain model was fitted to it by maximum likelihood estimation.  For each of the 1,000 data sets, we then measured the increase in deviance associated with removing higher order parameters

\begin{figure}
\begin{center}
\includegraphics[width=\columnwidth]{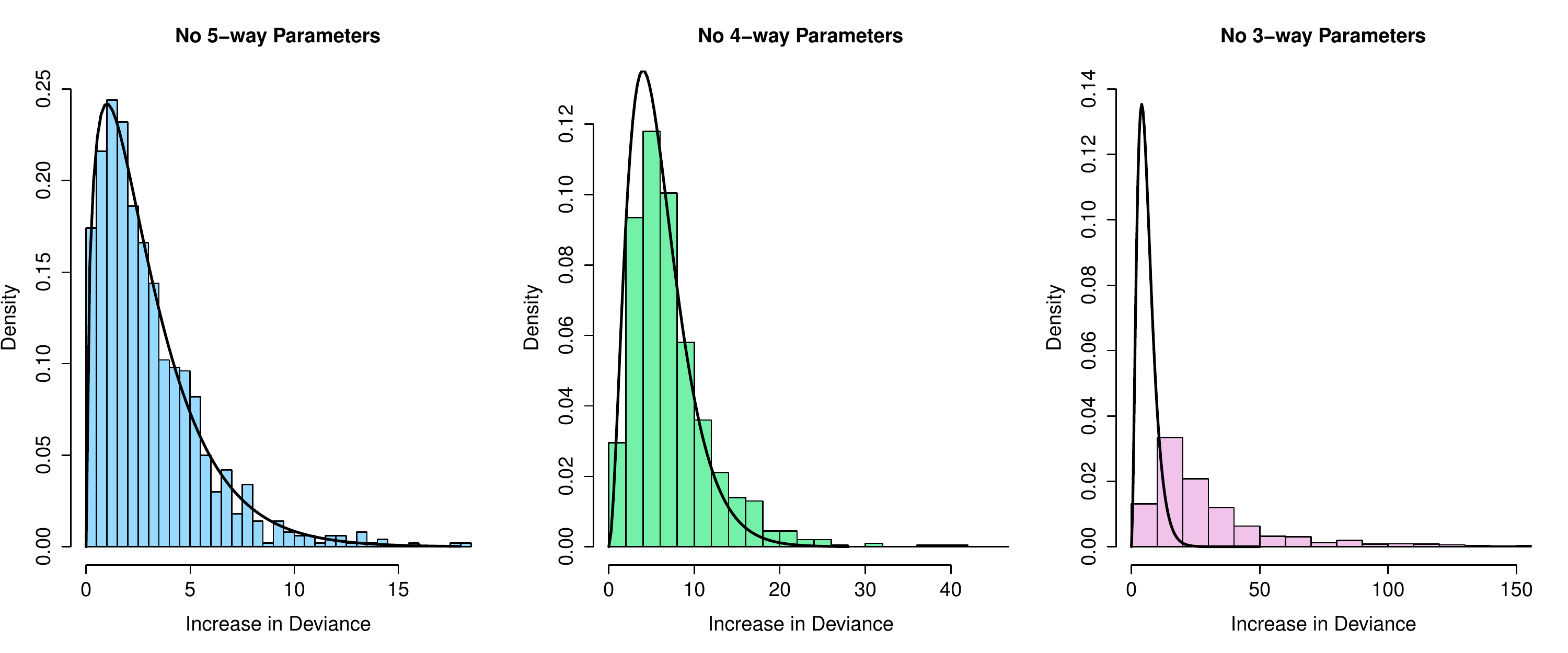}
\caption{Histograms showing the increase in deviance caused by setting to zero (a) the 5- and 6-way interaction parameters; (b) the 4-, 5- and 6-way interaction parameters; (c) the 3-, 4-, 5- and 6-way interaction parameters.  
Plots are based on $1,000$ datasets, each of size $10,000$, generated from
 the DAG in Figure \ref{fig:chain}(b).  The plotted densities are $\chi^2$ with 3, 6 and 10 degrees of freedom respectively.}
\label{fig:hists}
\end{center}
\end{figure}

The histogram in Figure \ref{fig:hists}(a) demonstrates that the deviance increase from setting the 5- and 6-way interaction parameters to zero (a total of three parameters) was not distinguishable from that which would be observed 
under the null hypothesis that these parameters are zero.  The deviance increase from setting the 4-, 5- and 6-way interactions to zero appeared to have only a slightly heavier tail than the associated $\chi^2$-distribution, as suggested by the outliers in Figure \ref{fig:hists}(b).  Removing the 3-way interactions in addition to this caused a dramatic increase in the deviance, as may be observed from the heavy tail of the histogram in Figure \ref{fig:hists}(c).
This illustrates that the ingenuous parametrization can be used to produce more parsimonious model descriptions than would be possible using Richardson's parameters.

Note that under the process which generated these models, each of these interaction parameters was non-zero almost surely.  As the sample size increases the power of a likelihood ratio test for a fixed distribution tends to one, so it must be the case that a simulation such as the above would, for large enough data sets, show significant deviation from the associated $\chi^2$ distributions.  However, even at a fairly large sample size of 10,000, a limited effect was observed in Figures \ref{fig:hists}(a) and (b), and the examples above with real data suggest that higher order interactions are often not particularly useful in practice for describing data.

\section{Proofs} \label{sec:proofs}

\subsection{Proof of Lemma \ref{lem:rbnnat}} \label{subsec:lem:rbnnat}

\begin{proof}[Proof of Lemma \ref{lem:rbnnat}]

Using the independence, we have \[p_{ABC}(\ii_{ABC}) = p_{AC}(\ii_{AC}) \cdot p_{B|C}(\ii_B \,|\, \ii_C).\]  Thus applying Lemma \ref{lem:mllp},
\begin{align*}
\lambda_{AD}^{ABC}(\ii_{AD}) = \frac{1}{|\X_{ABC}|} \sum_{\jj_{ABC} \in \X_{ABC}} \!\! (\log p_{AC}(\jj_{AC}) + \log p_{B|C}(\jj_B \,|\, \jj_C)) \prod_{v \in A \cup D} \!\! \left(|\X_v| \mathbb{I}_{\{x_v = y_v\}} - 1 \right).
\end{align*}
We can split this sum into terms involving $p_{AC}(\jj_{AC})$ and those involving $p_{B|C}(\jj_B \,|\, \jj_C)$.  For the first of these, 
\begin{align*}
\lefteqn{\frac{1}{|\X_{ABC}|} \sum_{\jj_{ABC} \in \X_{ABC}} \log p_{AC}(\jj_{AC}) \prod_{v \in A \cup D} \left(|\X_v| \mathbb{I}_{\{x_v = y_v\}} - 1 \right)}\\
&= \frac{1}{|\X_{AC}|\cdot|\X_B|} \sum_{\jj_{B} \in \X_{B}} \sum_{\jj_{AC} \in \X_{AC}} \log p_{AC}(\jj_{AC}) \prod_{v \in A \cup D} \left(|\X_v| \mathbb{I}_{\{x_v = y_v\}} - 1 \right)\\
&= \frac{1}{|\X_{AC}|} \sum_{\jj_{AC} \in \X_{AC}} \log p_{AC}(\jj_{AC}) \prod_{v \in A \cup D} \left(|\X_v| \mathbb{I}_{\{x_v = y_v\}} - 1 \right)\\
&= \lambda_{AD}^{AC}(\ii_{AC}),
\end{align*}
because the summand has no dependence on $\jj_B$.  For the latter,
\begin{align*}
\lefteqn{\frac{1}{|\X_{ABC}|} \sum_{\jj_{ABC} \in \X_{ABC}} \log p_{B|C}(\jj_B \,|\, \jj_C) \prod_{v \in A \cup D} \left(|\X_v| \mathbb{I}_{\{x_v = y_v\}} - 1 \right)}\\
&= \frac{1}{|\X_{ABC}|} \sum_{\jj_{BC} \in \X_{BC}} \log p_{B|C}(\jj_B \,|\, \jj_C) \sum_{\jj_A \in \X_A} \prod_{v \in A \cup D} \left(|\X_v| \mathbb{I}_{\{x_v = y_v\}} - 1 \right).
\end{align*}
Now for any $w \in A$, the inner part of this term is
\begin{align*}
\lefteqn{\sum_{\jj_A \in \X_A} \prod_{v \in A \cup D} \left(|\X_v| \mathbb{I}_{\{x_v = y_v\}} - 1 \right)}\\
&= \sum_{\jj_{A \setminus \{w\}}} \sum_{y_w} \prod_{v \in A \cup D} \left(|\X_v| \mathbb{I}_{\{x_v = y_v\}} - 1 \right)\\
&= \sum_{\jj_{A \setminus \{w\}}} \prod_{v \in (A \cup D) \setminus \{w\}} \left(|\X_v| \mathbb{I}_{\{x_v = y_v\}} - 1 \right) \sum_{y_w \in \X_w} \left(|\X_w| \mathbb{I}_{\{x_w = y_w\}} - 1 \right)\\
&= 0,
\end{align*}
because the innermost summand is $|\X_w| - 1$ for precisely one value of $y_w$, and $-1$ for the other $|\X_w| - 1$ values.  This shows that the whole term is zero, and gives the result.
\end{proof}

\subsection{Proof of Lemma \ref{lem:param}} \label{subsec:lem:param}

We first need the following result.

\begin{lemma} \label{lem:kappa}
For $L \subseteq M \subseteq V$ with $N \equiv M \setminus L$, define
\begin{align*}
\kappa_{L|N}(\ii_L \,|\, \ii_N) &\equiv \sum_{L \subseteq A \subseteq M} \lambda_A^M(\ii_A).
\end{align*}
Then
\begin{align*}
\kappa_{L|N}(\ii_L \,|\, \ii_N) &= \frac{1}{|\X_{L}|} \sum_{\substack{\jj_M \in \X_M \\ \jj_{N} = \ii_{N} }}  \log p(\jj_M) \prod_{v \in L} \left( |\X_{v}| \mathbb{I}_{\{x_v = y_v\}} - 1 \right).
\end{align*}
\end{lemma}

\begin{proof}
Applying Lemma \ref{lem:mllp}, we have
\begin{align*}
\lefteqn{\kappa_{L|N}(\ii_L \,|\, \ii_N)}\\
&= \sum_{L \subseteq A \subseteq M} \frac{1}{|\X_{M}|} \sum_{\jj_M \in \X_M}  \log p_M(\jj_M) \prod_{v \in A} \left( |\X_{v}| \mathbb{I}_{\{x_v = y_v\}} - 1 \right)\\
  &= \frac{1}{|\X_{M}|} \sum_{\jj_M \in \X_M}  \log p_M(\jj_M) \sum_{L \subseteq A \subseteq M} \prod_{v \in A} \left( |\X_{v}| \mathbb{I}_{\{x_v = y_v\}} - 1 \right)\\
  &= \frac{1}{|\X_{M}|} \sum_{\jj_M \in \X_M}  \log p_M(\jj_M) \sum_{L \subseteq A \subseteq M} \prod_{v \in L} \left( |\X_{v}| \mathbb{I}_{\{x_v = y_v\}} - 1 \right) \prod_{v \in A \setminus L} \left( |\X_{v}| \mathbb{I}_{\{x_v = y_v\}} - 1 \right)\\
  &= \frac{1}{|\X_{M}|} \sum_{\jj_M \in \X_M}  \log p_M(\jj_M) \prod_{v \in L} \left( |\X_{v}| \mathbb{I}_{\{x_v = y_v\}} - 1 \right) \sum_{B \subseteq N} \prod_{v \in B} \left( |\X_{v}| \mathbb{I}_{\{x_v = y_v\}} - 1 \right).
\end{align*}
Now, consider the value of the inner sum, for a fixed $\jj_M$.  In the case that there is some $w \in N$ with $x_w \neq y_w$, then
\begin{align*}
\sum_{B \subseteq N} \prod_{v \in B} \left( |\X_{v}| \mathbb{I}_{\{x_v = y_v\}} - 1 \right)
  &= \sum_{B \subseteq N \setminus \{w\}} \left[ \prod_{v \in B} \left( |\X_{v}| \mathbb{I}_{\{x_v = y_v\}} - 1 \right) + \!\!\! \prod_{v \in B \cup \{w\}} \!\!\! \left( |\X_{v}| \mathbb{I}_{\{x_v = y_v\}} - 1 \right)\right]\\
  &= \sum_{B \subseteq N \setminus \{w\}} \left[ \prod_{v \in B} \left( |\X_{v}| \mathbb{I}_{\{x_v = y_v\}} - 1 \right) - \prod_{v \in B} \left( |\X_{v}| \mathbb{I}_{\{x_v = y_v\}} - 1 \right)\right]\\
  &= 0.
\end{align*}
Alternatively, if $\ii_{N} = \jj_{N}$, then
\begin{align*}
\sum_{B \subseteq N} \prod_{v \in B} \left( |\X_{v}| \mathbb{I}_{\{x_v = y_v\}} - 1 \right)
  &= \sum_{B \subseteq N} \prod_{v \in B} \left( |\X_{v}| - 1 \right)\\
  &= |\X_{N}|
\end{align*}
by the binomial theorem.  Thus
\begin{align*}
\kappa_{L|N}(\ii_L \,|\, \ii_N) &= \frac{1}{|\X_{L}|} \sum_{\substack{\jj_M \in \X_M \\ \jj_{N} = \ii_{N} }}  \log p(\jj_M) \prod_{v \in L} \left( |\X_{v}| \mathbb{I}_{\{x_v = y_v\}} - 1 \right),
\end{align*}
since $\mathfrak{X}_M = \mathfrak{X}_L \times \mathfrak{X}_N$.
\end{proof}

\begin{proof}[Proof of Lemma \ref{lem:param}]
Let $N \equiv M \setminus L$, and pick some $\ii_L \in \tilde\X_L$ and $\ii_N \in \X_N$; for $A \subseteq L$, let $\boldsymbol 1_A$ be a vector of length $|L|$ with a 1 in position $j$ if the $j$th element of $L$ is in $A$, and 0 otherwise.
Define the local $|L|$-way log-linear interaction parameter between $\ii_L+\boldsymbol 1_L$ and $\ii_L$ conditional on $\ii_N$ as 
\begin{align*}
\sum_{A \subseteq L}  (-1)^{|L \setminus A|} \log p_{L|N}(\ii_L + \boldsymbol 1_A \,|\, \ii_{N});
\end{align*}
note that since $\ii_L \in \tilde\X_L$, $\ii_L+\boldsymbol 1_A \in \X_L$.
We will first show that we can construct all these local $|L|$-way log-linear interaction parameters using the parameters given in the statement of the lemma.
As in Lemma \ref{lem:kappa}, let $\kappa_{L|N}(\ii_L \,|\, \ii_N) \equiv \sum_{L \subseteq A \subseteq M} \lambda_A^M(\ii_A)$, and note that
\begin{align*}
\lefteqn{\sum_{A \subseteq L} (-1)^{|L \setminus A|} \kappa_{L|N}(\ii_L + \boldsymbol 1_A \,|\, \ii_N)}\\
&= \frac{(-1)^{|L|}}{|\X_{L}|} \sum_{\jj_L \in \X_L} \log p_M(\jj_L, \ii_N) \sum_{A \subseteq L} (-1)^{|A|} \prod_{v \in L} \left( |\X_{v}| \mathbb{I}_{\{x_v + \mathbb{I}_{\{v \in A\}} = y_v\}} - 1 \right)
\end{align*}
follows directly from Lemma \ref{lem:kappa}.  
Now consider the inner sum; if for some $w \in L$, $y_w \notin \{x_w, x_w+1\}$, then
\begin{align*}
&\sum_{A \subseteq L} (-1)^{|A|} \prod_{v \in L} \left( |\X_{v}| \mathbb{I}_{\{x_v + \mathbb{I}_{\{v \in A\}} = y_v\}} - 1 \right)\\ &\qquad = 
\sum_{A \subseteq L \setminus \{w\}} (-1)^{|A|} \left[ \prod_{v \in L} \left( |\X_{v}| \mathbb{I}_{\{x_v + \mathbb{I}_{\{v \in A\}} = y_v\}} - 1 \right) - \prod_{v \in L} \left( |\X_{v}| \mathbb{I}_{\{x_v + \mathbb{I}_{\{v \in A \cup\{w\}\}} = y_v\}} - 1 \right) \right]\\
&\qquad = 0,
\end{align*}
because the value of the outer indicator function is 0 in both terms when $v=w$, while the inner indicator functions are the same for all other $v$.  Alternatively,  if $y_w \in \{x_w, x_w+1\}$  for all $w \in L$, then define
\begin{align*}
B(A) \equiv \{v \in L \,|\, x_v + \mathbb{I}_{\{v \in A\}} = y_v \}.
\end{align*}
The map $A \mapsto B(A)$ is a one-to-one map from $\mathscr{P}(L)$, the power set of $L$, to itself, i.e.~an automorphism. Note that $D \equiv B(A) \triangle A = \{v\in L \,|\, x_v = y_v\}$ is independent of $A$.  Since
\[
|A| + 2|B(A) \setminus A| = |B(A)| + |A \triangle B(A)| = |B(A)|  + |D|
\]
 we can rewrite the sum over subsets as
\begin{align*}
&\sum_{A \subseteq L} (-1)^{|A|} \prod_{v \in L} \left( |\X_{v}| \mathbb{I}_{\{x_v + \mathbb{I}_{\{v \in A\}} = y_v\}} - 1 \right)\\
& \qquad = \sum_{A \subseteq L} (-1)^{|B(A)|+|D|} \prod_{v \in L} \left( |\X_{v}| \mathbb{I}_{\{v \in B(A)\}} - 1 \right)\\
& \qquad = (-1)^{|D|} \sum_{B \subseteq L} (-1)^{|B|} \prod_{v \in L} \left( |\X_{v}| \mathbb{I}_{\{v \in B\}} - 1 \right)\\
& \qquad = (-1)^{|D|} (-1)^{|L|} \sum_{B \subseteq L} \prod_{v \in B} \left( |\X_{v}| - 1 \right)\\
\intertext{which again using the binomial theorem is}
& \qquad = (-1)^{|D|} (-1)^{|L|} \prod_{v \in L} |\X_{v}| = (-1)^{|D|} (-1)^{|L|} |\X_L|.
\end{align*}
Then, substituting this back into the original expression and noting that the two $(-1)^{|L|}$ factors cancel out,
\begin{align*}
\sum_{A \subseteq L} (-1)^{|L \setminus A|} \kappa_{L|N}(\ii_L + \boldsymbol 1_A \,|\, \ii_N) &= \sum_{D \subseteq L}  (-1)^{|D|} \log p_M(\ii_L + \boldsymbol 1_{L\setminus D}, \, \ii_N)\\
  &= \sum_{D \subseteq L}  (-1)^{|D|} \left[ \log p_{L|N}(\ii_L + \boldsymbol 1_{L\setminus D} \,|\, \ii_{N}) + \log p_{N}(\ii_{N}) \right]\\
  &= \sum_{D \subseteq L}  (-1)^{|D|} \log p_{L|N}(\ii_L + \boldsymbol 1_{L\setminus D} \,|\, \ii_{N}),
\end{align*}
where the terms in $\log p_{N}(\ii_{N})$ cancel because of the lack of dependence upon $D$.  This is the (conditional) local $|L|$-way log-linear interaction.  The collection of all the (conditional) local $|L|$-way log-linear interactions together with the (conditional) $(|L|-1)$-dimensional marginal distributions smoothly parametrizes the $|L|$-way table \citep{csiszar:75, rudas:98}.
\end{proof}

\subsection{Proof of Theorem \ref{thm:subsp}} \label{subsec:thm:subsp}

We require the following lemma.

%

\begin{lemma} \label{lem:mb2}
Let $\bar{\G}$ be a head-preserving completion of ${\G}$, and let $H \in \mathcal{H}(\G)$ have tails $T$ and $\bar{T}$ in $\G$ and $\bar{\G}$ respectively.  Then under the global Markov property for $\G$,
\begin{align*}
H \indep (\bar{T} \setminus T)  \,|\, T \, [P].
\end{align*}
\end{lemma}

\begin{proof}
Let $\pi$ be a path in $\G$ from some $h \in H$ to $t \in \bar{T} \setminus T$, and assume without loss of generality that $\pi$ does not intersect $H$ or $\bar{T} \setminus T$ other than at its endpoints.  By Proposition \ref{prop:m-connecting-path}, every vertex on $\pi$ is in $\an_{\G}(\{h,t\}\cup T) \subseteq \an_{\G}(H \cup \bar{T})$.
Since $\bar{\G}$ is complete, if $v \in \an_{\bar{\G}}(H\cup \bar{T})$, then $v \in H\cup \bar{T}$, thus $H\cup \bar{T}$ is
ancestral in $\bar{\G}$. By Proposition \ref{prop:ancestors-in-completion}, $H\cup \bar{T}$ is also ancestral
in $\G$, thus every vertex on $\pi$ is in $H \cup \bar{T}$.

By Proposition \ref{prop:head-and-tail}, $\bar{T} \subseteq \an_{\bar{\G}}(H)$, so 
 $H\cup \bar{T} = \an_{\bar{\G}}(H)$. However, since $H$ forms a head in $\bar{\G}$, $H$  is barren in $\bar{\G}$.
 Thus in $\bar{\G}$, no proper descendant of a vertex in $H$ is on $\pi$, and by Proposition   \ref{prop:ancestors-in-completion} this also holds in $\G$.
 
 Now let $y$ be the first vertex after $h$ on $\pi$ that is not in $T$. By hypothesis, $y$ exists since $t \notin T$.
 By construction, any vertices between $h$ and $y$ on $\pi$ are in $T$, hence are colliders on $\pi$ and ancestors of $H$ in $\G$ (by Proposition \ref{prop:head-and-tail}). Thus $y \in \dis_{\G}(H) \cup \pa_{\G}(\dis_{\G}(H))$.  If $y \in \an_{\G}(H)$ then $y \in T$, which is a contradiction, hence $y \in \dis_{\G}(H)$ and $y \notin\an_{\G}(H)$. As shown earlier,
 $y$ is not a descendant of a vertex in $H$, so $H\cup \{y\}$ forms a head in $\G$. Since $\bar{\G}$ is a head-preserving completion, it follows that $H\cup \{y\}$ also forms a head in $\bar{\G}$, and thus $y \notin \an_{\bar{\G}}(H) = H \cup \bar{T}$, but this is a contradiction.
 \end{proof}

\begin{proof}[Proof of Theorem \ref{thm:subsp}]
Let $(H,\bar{T})$ be a head-tail pair in $\bar{\G}$.  There are three possibilities for how this pair relates to $\G$: if $(H,\bar{T})$ is also a head-tail pair in $\G$, then there is no work to be done; otherwise either (i) $H$ is not a head in $\G$, or (ii) $H$ is a head in $\G$ but $\bar{T}$ is not its tail.

If (i) holds, then we claim that under $\G$, $\lambda_{A}^{H \bar{T}} = 0$ for all $H \subseteq A \subseteq H \cup \bar{T}$.  To see this, first note that $H$ is a barren set in $\bar{\G}$, and since $H$ is maximally connected, this means that all elements are joined by bidirected edges in $\bar{\G}$. Since $\G$ contains a subset of the edges in $\bar{\G}$,
  $H$ is also barren in $\G$;  since $H$ is not a head in $\G$ this means that $H = K \cup L$ for disjoint non-empty sets $K$ and $L$ with no edges directly connecting them.  But this implies that $K$ and $L$ are m-separated conditional on $\bar{T}$, and thus $X_K \indep X_L \,|\, X_{\bar{T}}$ under the Markov property for $\G$.  Then, by Lemma \ref{lem:brindep}, these parameters are all identically zero under $\G$.

(ii) implies that $H$ is head in both $\G$ and $\bar{\G}$, but $\bar{T} \equiv \tail_{\bar{\G}} (H) \supset \tail_{\G} (H) \equiv T$.
Then $\lambda_A^{H\bar{T}} = 0$ for all $H \subseteq A \subseteq H \cup \bar{T}$ such that $A \cap (\bar{T} \setminus T) \neq \emptyset$; this follows from Lemma \ref{lem:mb2}
and application of Lemma \ref{lem:brindep}.

We have shown that all parameters corresponding to effects not found in $\PP^{\ing}(\G)$ are identically zero under $\G$.  The vanishing of these parameters defines the correct sub-model, but note that some of the margins in $\PP^{\ing}(\bar{\G})$ which we have not yet considered are not the same as those in $\PP^{\ing}(\G)$.  These remaining cases are again from (ii), but where $H \subseteq A \subseteq H \cup T$; in this case $\lambda_{A}^{H\bar{T}} = \lambda_{A}^{HT}$ under $\G$, again due to Lemma \ref{lem:mb2}, this time combined with Lemma \ref{lem:rbnnat}.  

Thus we have shown that under $\G$, all the ingenuous parameters for $\bar\G$ are either zero or equal to ingenuous parameters for $\G$.  Combined with Theorem \ref{thm:param}, this shows that those constraints define the model.
\end{proof}

\subsection{Proof of Theorem \ref{thm:ingVI}} \label{subsec:thm:ingVI}

We first prove the following graphical result.

\begin{lemma} \label{lem:heads}
Let $\G$ be an ADMG containing at least one head of size 3 or more.  Then $\G$ also contains two heads of the form $\{v_1, v_2\}$ and $\{v_2, v_3\}$, where $\{v_1, v_2, v_3\}$ is barren.
\end{lemma}

\begin{proof}
Suppose not; let $\G$ be an ADMG which violates this condition, and let $H$ be a head in $\G$ of size $k \geq 3$.  Pick 3 vertices $\{w_1, w_2, w_3\}$ in $H$.  By the definition of a head, we can pick a bidirected path $\pi$, through $\an_{\G} (H)$, from $w_1$ to $w_2$; assume that $\pi$ contains no other element of $H$, otherwise shorten the path and redefine $w_1$ or $w_2$.  Then create a similar path $\rho$ from $w_2$ to $w_3$; again assume that $\rho$ contains no other element of $H$, else shorten the path and redefine $w_3$.  If $w_1$ lies on $\rho$ then we can swap $w_1$ and $w_2$ to get the desired result.

According to our assumption that the result is false, at least one of $\{w_1, w_2\}$ or $\{w_2, w_3\}$ is not a head; assume the former without loss of generality.  This implies that $\pi$ must pass through at least one vertex $v$ which is not an ancestor of $\{w_1, w_2\}$.  If there is more than one such vertex, then choose one which has no distinct descendants on the path $\pi$.  By the construction of $\pi$ we have $v \in \an_{\G} (H) \setminus H$.

Then let $W$ be the set of vertices on $\pi$, and $H^* \equiv \barren_{\G} (W)$.  Since $W$ is $\leftrightarrow$-connected, $H^*$ must be a head, and $\{w_1, w_2, v\} \subseteq H^*$.  Thus we have created a head distinct from $H$, of size at least 3, which is contained in the set of ancestors of $H$.

The assumption we have made implies that we must be able to repeat this process indefinitely, with each head being contained in the ancestors of the previous head.  To see that we never obtain the same head twice, note that there is a non-empty directed path from $v \in H^*$ to $H$; but $H$ is contained within the ancestors of any previous heads in the sequence, so if $H^*$ had appeared before, this would imply that $H^*$ was not barren.

Then since $H$ has a finite set of ancestors, the apparently infinite recursion of distinct heads is a contradiction.
\end{proof}

\begin{definition}
Let $A$ be an ancestral set in an ADMG $\G$, and let $v \in \barren_\G(A)$.  The \emph{Markov blanket} for $v$ in $A$ is the set
\begin{align*}
\mbl(v, A) \equiv \pa_A(\dis_A(v)) \cup (\dis_A(v) \setminus \{v\}).
\end{align*}
In particular, under the ordered local Markov property for $\G$,
\begin{align}\label{eq:ordered-local-markov}
v \indep A \setminus (\mbl(v, A) \cup \{v\}) \,|\, \mbl(v, A).
\end{align}
Note that (\ref{eq:ordered-local-markov}) holds for every $v$ and ancestral set $A$ (with  $v \in \barren_\G(A)$) if and only if the global Markov property for $\G$ holds \citep{richardson:03}.
\end{definition}

\begin{proof}[Proof of Theorem \ref{thm:ingVI}]
($\Leftarrow$).  Suppose that $\G$ contains no heads of size $\geq 3$, and let $1, \ldots, n$ be a topological ordering on the vertices of $\G$.  We will construct a complete, hierarchical and variation independent parametrization of the saturated model, and then show that under the global Markov property for $\G$ it is equivalent to the ingenuous parametrization.

Let $\mathbb{M}_i \subseteq \mathbb{M}$ be the margins which involve only the vertices in $[i] = \{1, \ldots i\}$.  Assume for induction, that $\mathbb{M}_{i-1}$ includes the set $[i-1]$, and these margins and their associated effects are hierarchical, complete and satisfy the ordered decomposability criterion up to this point.  The base case for $i=1$ is trivial.

Now, let the heads involving $i$ contained within $[i]$ be $H_0 = \{i\}, H_1 = \{j_1, i\}, \ldots, H_k = \{j_k, i\}$, where $j_1 < \ldots < j_k < i$ (possibly with $k=0$).  Call the associated tails $T_0, \ldots, T_k$.  We have
\[
\barren_{\G} \left( \dis_{\G} (i) \right) = \{j_k, i\},
\]
 since $\barren_{\G} \left( \dis_{\G} (i) \right)$ is a head, and cannot have size $\geq 3$.  This also implies that $(H_k \cup T_k) \setminus \{i\} = \mbl(i, [i])$, where $\mbl(v, A)$ is the Markov blanket of $v$ in the ancestral set $A$.

Now, since the ordering is topological, $A_k \equiv [i]$ is an ancestral set, and the ordered local Markov property shows that
\begin{align*}
i \indep A_k \setminus (\mbl(i, A_k) \cup \{i\}) \;|\; \mbl(i, A_k),
\end{align*}
so
\begin{align*}
i \indep A_k \setminus (H_k \cup T_k) \;|\; (H_k \cup T_k) \setminus \{i\}.
\end{align*}
Then for all $\{i\} \subseteq C \subseteq A_k$ such that $C \cap \dec_\G(j_k) \neq \emptyset$,
\begin{align*}
\lambda_C^{A_k} &= \lambda_C^{H_k \cup T_k} && \mbox{if } H_k \subseteq C \subseteq H_k \cup T_k\\
\lambda_C^{A_k} &= 0 && \mbox{otherwise},
\end{align*}
where the first equality follows from the independence and Lemma \ref{lem:rbnnat}, and the second from the above independence and Lemma \ref{lem:brindep}.


Now set $A_{k-1} = A_k \setminus \dec_{\G} (j_k)$.  Then $A_{k-1}$ is ancestral and contains $i$, so applying the ordered local Markov property again gives for any $\{i\} \subseteq C \subseteq A_{k-1}$ such that $C \cap \dec_\G(j_{k-1}) \neq \emptyset$,
\begin{align*}
\lambda_C^{A_{k-1}} &= \lambda_C^{H_{k-1} \cup T_{k-1}} && \mbox{if } H_{k-1} \subseteq C \subseteq H_{k-1} \cup T_{k-1}\\
\lambda_C^{A_{k-1}} &= 0 && \mbox{otherwise}.
\end{align*}
Continuing this approach gives exactly one parameter for each subset $C$ of $[i]$ containing $i$ and some descendant of any of $j_1, \ldots, j_k$.  Lastly let $A_0 = A_1 \setminus \dec_{\G} (j_1)$.  Then for $\{i\} \subseteq C \subseteq A_0$, 
\begin{align*}
\lambda_C^{A_0} &= \lambda_C^{H_0 \cup T_0} && \mbox{if } \{i\} \subseteq C \subseteq \{i\} \cup T_0\\
\lambda_C^{A_0} &= 0 && \mbox{otherwise.}
\end{align*}


Now, add the margins $A_0 \subset \cdots \subset A_k = [i]$; since these all contain $\{i\}$, they are not a subset of any existing margin.  Further, each set $C$ we associate with $A_l$ contains a vertex which is not in $A_{l-1}$.  Thus the addition of these margins and their associated effects keeps our parametrization complete and hierarchical.  Setting $\mathbb{M}_{i} = \mathbb{M}_{i-1} \cup \{A_0, \ldots, A_k\}$, then there are at most two maximal subsets out of the margins up to $A_l$ (being $[i-1]$ and $A_l$); thus $\mathbb{M}_i$ is clearly also ordered decomposable, and so the parameters are variation independent.  

Furthermore we have shown that under the global Markov property for $\G$, these parameters are equal to the ingenuous parameters or are identically zero.  Thus the ingenuous parameters must also be variation independent.

($\Rightarrow$).  Our construction will assume the random variables are binary; the general case is a trivial but tedious extension.  Suppose that $\G$ has a head of size $\geq 3$, and assume for a contradiction that its ingenuous parametrization is variation independent.  Then by Lemma \ref{lem:heads}, there exist two heads $H_1 = \{v_1, v_2\}$ and $H_2 = \{v_2, v_3\}$ such that $\{v_1, v_2, v_3\}$ is barren.  Let $H_3 \equiv \{v_3, v_1\}$ noting that this set may or may not be a head.

Also let $T_i = \tail_{\G} (H_i)$, where if $H_3$ is not a head, this set is taken to be the tail of $H_3$ \emph{if there were} a bidirected arrow between $v_1$ and $v_3$.  Further let $A = \an_{\G} (H)$.

Now choose $\lambda_{C_i}^{B_i} = 0$, where $B_i = \{v_i\} \cup \tail_{\G} (v_i)$ and $\{v_i\} \subseteq C_i \subseteq B_i$; this sets every $v_i$ to be uniform on $\{0,1\}$ for each instantiation of its tail.

Similarly, by choosing $\lambda_{C_1}^{H_1 \cup T_1}(0)$ to be large and positive for each $H_1 \subseteq C_1 \subseteq H_1 \cup T_1$, we can force $v_1$ and $v_2$ to be arbitrarily highly correlated conditional on $T_1$, and therefore conditional on $A$.  We can do the same for $v_2$ and $v_3$, so for any $0 < \epsilon <\tfrac{1}{2}$:
\begin{multicols}{2}
\begin{center}
\begin{tabular}{c|c|cc|}
\multicolumn{2}{c}{}&\multicolumn{2}{c}{$v_1$}\\
\cline{2-4}
&&0&1\\
\cline{2-4}
\multirow{2}{*}{$v_2$}&0&$\frac{1}{2}-\epsilon$&$\epsilon$\\
&1&$\epsilon$&$\frac{1}{2}-\epsilon$\\
\cline{2-4}
\end{tabular}

\begin{tabular}{c|c|cc|}
\multicolumn{2}{c}{}&\multicolumn{2}{c}{$v_2$}\\
\cline{2-4}
&&0&1\\
\cline{2-4}
\multirow{2}{*}{$v_3$}&0&$\frac{1}{2}-\epsilon$&$\epsilon$\\
&1&$\epsilon$&$\frac{1}{2}-\epsilon$\\
\cline{2-4}
\end{tabular},
\end{center}
\end{multicols}
where these tables are understood to show the two-way marginal distributions conditional on each instantiation $\ii_A$ of $A$.

But now either $\lambda_{C_3}^{H_3 \cup T_3} = 0$ by design (because $H_3$ is not a head, and $v_1$ and $v_3$ are independent conditional on their `tail'), or we can choose this to be the case by the assumption of variation independence.  This implies that $v_1$ and $v_3$ are independent conditional on $A$.  Thus
\begin{align*}
\frac{1}{4} &= P(v_1 = 1, v_3 = 0 \,|\, A = \ii_A)\\
 &= P(v_1 = 1, v_2 = 0, v_3 = 0 \,|\, A = \ii_A) + P(v_1 = 1, v_2 = 1, v_3 = 0 \,|\, A = \ii_A)\\
&< P(v_1 = 1, v_2 = 0 \,|\, A = \ii_A) + P(v_2 = 1, v_3 = 0 \,|\, A = \ii_A)\\
&= 2 \epsilon,
\end{align*}
which is a contradiction if $\epsilon < \tfrac{1}{8}$.  Thus the parameters are variation dependent.
\end{proof}

\subsection*{Acknowledgements}

This research was supported by the U.S.\ National Science Foundation grant CNS-0855230 and U.S.\ National Institutes of Health grant R01 AI032475.  
The Netherlands Kinship Panel Study is funded by grant 480-10-009 from the Major Investments Fund of the Netherlands Organisation for Scientific Research (NWO), and by the Netherlands Interdisciplinary Demographic Institute (NIDI), Utrecht University, the University of Amsterdam and Tilburg University. We thank McDonald, Hiu and Tierney for giving us permission to use their flu vaccine data.

Our thanks go to Tam\'as Rudas for helpful discussions, and to Antonio Forcina for discussions and the use of his computer programmes.  Finally we thank two anonymous referees and an associate editor for their thorough reading of an earlier draft, and very useful suggestions.

\bibliographystyle{plainnat}
\bibliography{mybib}

\begin{thebibliography}{32}
\providecommand{\natexlab}[1]{#1}
\providecommand{\url}[1]{\texttt{#1}}
\expandafter\ifx\csname urlstyle\endcsname\relax
  \providecommand{\doi}[1]{doi: #1}\else
  \providecommand{\doi}{doi: \begingroup \urlstyle{rm}\Url}\fi

\bibitem[Ackley et~al.(1985)Ackley, Hinton, and Sejnowski]{ackley:85}
D.~H. Ackley, G.~E. Hinton, and T.~J. Sejnowski.
\newblock A learning algorithm for {B}oltzmann machines.
\newblock \emph{Cognitive Science}, 9:\penalty0 147--169, 1985.

\bibitem[Ali et~al.(2005)Ali, Richardson, Spirtes, and Zhang]{ali:05}
R.~A. Ali, T.~S. Richardson, P.~Spirtes, and J.~Zhang.
\newblock Towards characterizing {M}arkov equivalence classes for directed
  acyclic graphs with latent variables.
\newblock In \emph{Proceedings of the 21st Conference on Uncertainty in
  Artificial Intelligence}, pages 10--17, 2005.

\bibitem[Bergsma and Rudas(2002)]{br:02}
W.~P. Bergsma and T.~Rudas.
\newblock Marginal models for categorical data.
\newblock \emph{Ann. Stat.}, 30\penalty0 (1):\penalty0 140--159, 2002.

\bibitem[Csisz\'ar(1975)]{csiszar:75}
I.~Csisz\'ar.
\newblock $i$-divergence geometry of probability distributions and minimization
  problems.
\newblock \emph{Annals of Probability}, 3\penalty0 (1):\penalty0 146--158,
  1975.

\bibitem[Darroch et~al.(1980)Darroch, Lauritzen, and Speed]{darroch:80}
J.~N. Darroch, S.~L. Lauritzen, and T.~P. Speed.
\newblock Markov fields and log-linear models for contingency tables.
\newblock \emph{Ann. Statist.}, 8:\penalty0 522--539, 1980.

\bibitem[Dawid(1979)]{dawid:condind}
A.~P. Dawid.
\newblock Conditional independence in statistical theory (with discussion).
\newblock \emph{J.\ Roy.\ Statist.\ Soc.\ Ser.\ B}, 41:\penalty0 1--31, 1979.

\bibitem[Drton(2009)]{drton:09}
M.~Drton.
\newblock Discrete chain graph models.
\newblock \emph{Bernoulli}, 15\penalty0 (3):\penalty0 736–--753, 2009.

\bibitem[Drton and Richardson(2008{\natexlab{a}})]{drton:richardson:08}
M.~Drton and T.~S. Richardson.
\newblock Binary models for marginal independence.
\newblock \emph{J. Roy. Statist. Soc. Ser. B}, 70\penalty0 (2):\penalty0
  287--309, 2008{\natexlab{a}}.

\bibitem[Drton and Richardson(2008{\natexlab{b}})]{drton:richardson:08a}
M.~Drton and T.~S. Richardson.
\newblock Graphical methods for efficient likelihood inference in {G}aussian
  covariance models.
\newblock \emph{J.\ Mach.\ Learn.\ Res.}, 9:\penalty0 893--914,
  2008{\natexlab{b}}.

\bibitem[Dykstra et~al.(2005)Dykstra, Kalmijn, Knijn, Komter, Liefbroer, and
  Mulder]{NKPS:05}
P.~A. Dykstra, M.~Kalmijn, T.~C.~M. Knijn, A.~E. Komter, A.~C. Liefbroer, and
  C.~H. Mulder.
\newblock Codebook of the {N}etherlands {K}inship {P}anel {S}tudy, a
  multi-actor, multi-method panel study on solidarity in family relationships.
  {W}ave 1.
\newblock \emph{NKPS Working Paper No.\ 4}, 2005.

\bibitem[Dykstra et~al.(2007)Dykstra, Kalmijn, Knijn, Komter, Liefbroer, and
  Mulder]{NKPS:07}
P.~A. Dykstra, M.~Kalmijn, T.~C.~M. Knijn, A.~E. Komter, A.~C. Liefbroer, and
  C.~H. Mulder.
\newblock Codebook of the {N}etherlands {K}inship {P}anel {S}tudy, a
  multi-actor, multi-method panel study on solidarity in family relationships.
  {W}ave 2.
\newblock \emph{NKPS Working Paper No.\ 6}, 2007.

\bibitem[Ekholm et~al.(2012)Ekholm, Jokinen, McDonald, and Smith]{ekholm:12}
A.~Ekholm, J.~Jokinen, J.W. McDonald, and P.W.F. Smith.
\newblock A latent class model for bivariate binary responses from twins.
\newblock \emph{J.\ Roy.\ Statist.\ Soc.\ Ser.\ C}, 61\penalty0 (3), 2012.

\bibitem[Evans and Forcina(2011)]{evans:11}
R.~J. Evans and A.~Forcina.
\newblock Two algorithms for fitting constrained marginal models.
\newblock Arxiv preprint arXiv:1110.2894, 2011.

\bibitem[Evans and Richardson(2010)]{evans:10}
R.~J. Evans and T.~S. Richardson.
\newblock Maximum likelihood fitting of acyclic directed mixed graphs to binary
  data.
\newblock In \emph{Proceedings of the 26th conference on Uncertainty in
  Artificial Intelligence}, 2010.

\bibitem[Forcina et~al.(2010)Forcina, Lupparelli, and Marchetti]{forcina:10}
A.~Forcina, M.~Lupparelli, and G.~M. Marchetti.
\newblock Marginal parameterizations of discrete models defined by a set of
  conditional independencies.
\newblock \emph{Journal of Multivariate Analysis}, 101:\penalty0 2519--2527,
  2010.

\bibitem[Glonek and McCullagh(1995)]{glonek:95}
G.~F.~V. Glonek and P.~McCullagh.
\newblock Multivariate logistic models.
\newblock \emph{J.\ Roy.\ Statist.\ Soc.\ Ser.\ B}, 57\penalty0 (3):\penalty0
  533--546, 1995.

\bibitem[Hakim et~al.(2003)Hakim, Cherkas, Spector, and MacGregor]{hakim:03}
A.~J. Hakim, L.~F. Cherkas, T.~D. Spector, and A.~J. MacGregor.
\newblock Genetic associations between frozen shoulder and tennis elbow: a
  female twin study.
\newblock \emph{Rheumatology}, 42\penalty0 (6):\penalty0 739--742, 2003.

\bibitem[Kauermann(1997)]{kauermann:97}
G.~Kauermann.
\newblock A note on multivariate logistic models for contingency tables.
\newblock \emph{Austral. J. Statist.}, 39\penalty0 (3):\penalty0 261--276,
  1997.

\bibitem[Lauritzen(1996)]{lau:96}
S.~L. Lauritzen.
\newblock \emph{Graphical Models}.
\newblock Clarendon Press, Oxford, UK, 1996.

\bibitem[Lupparelli et~al.(2009)Lupparelli, Marchetti, and
  Bergsma]{lupparelli:09}
M.~Lupparelli, G.~M. Marchetti, and W.~P. Bergsma.
\newblock Parameterizations and fitting of bi-directed graph models to
  categorical data.
\newblock \emph{Scand. J. Statist.}, 36:\penalty0 559--576, 2009.

\bibitem[Marchetti and Lupparelli(2011)]{marchetti:11}
G.~M. Marchetti and M.~Lupparelli.
\newblock Chain graph models of multivariate regression type for categorical
  data.
\newblock \emph{Bernoulli}, 17\penalty0 (3):\penalty0 827--844, 2011.

\bibitem[McDonald et~al.(1992)McDonald, Hui, and Tierney]{mcdonald:92}
C.~J. McDonald, S.~L. Hui, and W.~M. Tierney.
\newblock Effects of computer reminders for influenza vaccination on morbidity
  during influenza epidemics.
\newblock \emph{MD computing}, 9\penalty0 (5):\penalty0 304, 1992.

\bibitem[Pearl(1988)]{pearl:88}
J.~Pearl.
\newblock \emph{Probabilistic Reasoning in Intelligent Systems}.
\newblock Morgan Kaufmann, 1988.

\bibitem[Richardson(2003)]{richardson:03}
T.~S. Richardson.
\newblock Markov properties for acyclic directed mixed graphs.
\newblock \emph{Scand. J. Statist.}, 30\penalty0 (1):\penalty0 145--157, 2003.

\bibitem[Richardson(2009)]{richardson:09}
T.~S. Richardson.
\newblock A factorization criterion for acyclic directed mixed graphs.
\newblock In \emph{Proceedings of the 25th conference on Uncertainty in
  Artificial Intelligence}, 2009.

\bibitem[Richardson and Spirtes(2002)]{richardson:02}
T.~S. Richardson and P.~Spirtes.
\newblock Ancestral graph {M}arkov models.
\newblock \emph{Ann.~Statist.}, 30:\penalty0 962--1030, 2002.

\bibitem[Rudas(1998)]{rudas:98}
T.~Rudas.
\newblock \emph{Odds ratios in the analysis of contingency tables}.
\newblock Sage Publications, Inc, 1998.

\bibitem[Rudas et~al.(2006)Rudas, Bergsma, and N\'emeth]{rudas:06}
T.~Rudas, W.~P. Bergsma, and R.~N\'emeth.
\newblock Parameterization and estimation of path models for categorical data.
\newblock In \emph{Proceedings in Computational Statistics, 17th Symposium},
  pages 383--394. Physica-Verlag HD, 2006.

\bibitem[Rudas et~al.(2010)Rudas, Bergsma, and N\'emeth]{rudas:10}
T.~Rudas, W.~P. Bergsma, and R.~N\'emeth.
\newblock Marginal log-linear parameterization of conditional independence
  models.
\newblock \emph{Biometrika}, 97:\penalty0 1006--1012, 2010.

\bibitem[Siemiatycki(1979)]{siemiatycki:79}
J.~Siemiatycki.
\newblock A comparison of mail, telephone, and home interview strategies for
  household health surveys.
\newblock \emph{American Journal of Public Health}, 69:\penalty0 238--245,
  1979.

\bibitem[Wermuth(2011)]{wermuth:11}
N.~Wermuth.
\newblock Probability distributions with summary graph structure.
\newblock \emph{Bernoulli}, 17\penalty0 (3):\penalty0 845--879, 2011.

\bibitem[Whittaker(1990)]{whittaker:90}
J.~Whittaker.
\newblock \emph{Graphical models in applied multivariate statistics}.
\newblock Wiley, 1990.

\end{thebibliography}

\end{document}